\documentclass[a4paper]{article}

\usepackage[pages=all, color=black, position={current page.south}, placement=bottom, scale=1, opacity=1, vshift=5mm]{background}
\SetBgContents{}      % copyright

\usepackage[margin=1in]{geometry} % full-width

\usepackage{graphicx}

\usepackage{amsmath}
\usepackage{amssymb}
\usepackage{amsthm}
\usepackage{marvosym}
\usepackage[scr=boondox]{mathalpha}

\usepackage{algorithm}
\usepackage{algpseudocode}

\usepackage{tikz}
\usetikzlibrary{automata,positioning}
\usepackage[utf8]{inputenc} % Necessario per il carattere 'Č'
\usepackage[T1]{fontenc}    % Migliore codifica dei font
\usetikzlibrary{shapes.geometric, fit, positioning}
\usetikzlibrary{shadows.blur}

\usepackage{amsmath}
\usepackage{amsthm}
\usepackage{amssymb}

\usepackage[utf8]{inputenc}
\usepackage{hyperref}
\hypersetup{
	unicode,
	pdfauthor={Venturi Riccardo, Emanuele Rodaro},
	pdftitle={The hereditariness problem for the Černý conjecture},
	pdfsubject={The hereditariness problem for the Černý conjecture},
	pdfkeywords={article, template, simple},
	pdfproducer={LaTeX},
	pdfcreator={pdflatex}
}

\theoremstyle{theorem}
\newtheorem{teo}{Theorem}[section]
\newtheorem{lemma}[teo]{Lemma}
\newtheorem{cor}[teo]{Corollary}
\newtheorem{conj}[teo]{Conjecture}
\newtheorem{opbm}[teo]{Open Problem}
\newtheorem{prop}[teo]{Proposition}

\newtheorem*{teo*}{Theorem}

\theoremstyle{definition}
\newtheorem{defi}[teo]{Definition}

\theoremstyle{remark}
\newtheorem{oss}[teo]{Remark}
\newtheorem{ese}[teo]{Example}
\newcommand{\IdM}{\operatorname{Id}}

\usepackage{graphicx, color}
\graphicspath{{fig/}}

\usepackage{algorithm, algpseudocode} % use algorithm and algorithmicx for typesetting algorithms
\usepackage{mathrsfs} % for \mathscr command

\usepackage{lipsum}

\newcommand{\A}{\mathcal{A}}
\newcommand{\M}{\text{M}}

\newcommand{\Syn}{\text{Syn}}
\newcommand{\rad}{\text{Rad}}
\newcommand{\Con}{\text{Cong}}

\title{The hereditariness problem for the Černý conjecture}
\author{
 Emanuele Rodaro
\footnote{ 
Politecnico di Milano, Milan, Italy.
\texttt{emanuele.rodaro@polimi.it}}
\and 
 Riccardo Venturi
\footnote{Universidade NOVA de Lisboa, Lisbon, Portugal.
\texttt{r.venturi@campus.fct.unl.pt}}
}

\usepackage{csquotes}
\usepackage[
  style=numeric,
  sorting=nyt,
  sortcites=true,
  maxnames=5,
  minnames=4
]{biblatex}
\begin{document}

\maketitle
	
\begin{abstract}
This paper addresses the lifting problem for the \v{C}ern\'y conjecture: namely, whether the validity of the conjecture for a quotient automaton can always be transferred (or “lifted”) to the original automaton. Although a complete solution remains open, we show that it is sufficient to verify the \v{C}ern\'y conjecture for three specific subclasses of reset automata: \emph{radical}, \emph{simple}, and \emph{quasi-simple}.  
Our approach relies on establishing a Galois connection between the lattices of congruences and ideals of the transition monoid. This connection not only serves as the main tool in our proofs but also provides a systematic method for computing the radical ideal and for deriving structural insights about these classes.  
In particular, we show that for every simple or quasi-simple automaton $\A$, the transition monoid $\M(\A)$ possesses a unique ideal covering the minimal ideal of constant (reset) maps; a result of similar flavor holds for the class of radical automata.
\end{abstract}



\section{Introduction}
A complete deterministic semiautomaton is an action of the free monoid $\Sigma^*$ on a finite set $Q$. More concretely, it can be described as a tuple $\A = (Q, \Sigma, \delta)$, where $Q$ is a finite set of states, $\Sigma$ is a finite alphabet, and $\delta: Q \times \Sigma \to Q$ is the transition function. This function specifies the action of the alphabet $\Sigma$ on the state set $Q$ and extends naturally to words in $\Sigma^*$. Combinatorially, a semiautomaton can be viewed as a directed, edge-labelled graph in which each vertex has exactly one outgoing edge labelled by each $a \in \Sigma$. A complete deterministic semiautomaton is commonly referred to simply as an \textbf{automaton}, or sometimes as a \textbf{DFA}, and is often used in theoretical computer science to recognize languages when an initial state and a set of final states are specified. However, our focus is on automata from a combinatorial and algebraic perspective. Our primary motivation stems from the \v{C}ern\'y conjecture, one of the longest-standing open problems in automata theory, which is still unsolved after more than sixty years. An automaton is called \textbf{synchronizing} (or \textbf{reset}) if there exists a word $w \in \Sigma^{*}$, referred to as a \emph{synchronizing} (or \emph{reset}) word, that maps all states to a single, identical state. Formally, for a synchronizing word $w$, we have $q \cdot w = q' \cdot w$ for all states $q,q' \in Q$. The \v{C}ern\'y conjecture asserts that any synchronizing automaton with $n$ states admits a synchronizing word of length at most $(n-1)^2$ (see \cite{Ce64}). 
The study of this conjecture has generated an extensive body of work investigating various aspects of synchronizing automata. This research addresses algorithmic methods for finding (``short") reset words, establishes proofs of the conjecture for specific classes of automata, and explores connections with other areas such as the theory of codes and symbolic dynamics, most notably through the Road Coloring Problem. See, for instance, \cite{AlRo, BeBePe, Dubuc, Epp, Kari, Ro18Adv, Steinb, Trah, Trah_road_coloring, Vo_CIAA07, Don, BehJoh, Perrin_unamb_coded_shifts}. For a survey of all possible subclasses of automata for which the \v{C}ern\'y conjecture has been solved, see \cite{survey_volkov_2025}. The bound $(n-1)^2$ is known to be tight, as it is attained by the \v{C}ern\'y automata—the only known infinite family achieving this bound—along with a few other sporadic examples (see \cite{Vo_Survey}). An automaton whose shortest reset word attains the $(n-1)^2$ bound is called an \textbf{extremal automaton}. Our main result proves that the existence of extremal automata is confined to significantly more restricted classes than 
previously known.
Regarding the best known upper bound for the reset threshold, a general  cubic bound on the length of synchronizing words was first established in the seminal works of Frankl and Pin~\cite{Frankl, Pin}. Subsequent research has led to improvements on this bound, with notable contributions for the general case by Shitov and Szykula \cite{Shitov, Szy17}, and for the specific class of semisimple synchronizing automata in \cite{RoPack18}. For a comprehensive overview of synchronizing automata and the  \v{C}ern\'y conjecture, we refer the reader to several excellent surveys, notably those by Kari and Volkov~\cite{KaVo} and by Volkov~\cite{Vo_Survey, VolkSurv2}.
\\
Among the various approaches to this conjecture, representation theory of monoids plays a particularly important role. 
For further reading on this topic, see the works of Steinberg~\cite{Steinb, SteinbEJC}, Almeida et al.~\cite{AlMaStVo}, B\'eal and Perrin~\cite{BeBePe}, 
Dubuc~\cite{Dubuc}, and Arnold and Steinberg~\cite{ArnSteinb}, among others. In this direction, the papers \cite{AlRo,RoPack18} explore \v{C}ern\'y's conjecture from a ring-theoretic perspective. In \cite[Corollary~9]{AlRo}, it is shown that if an automaton contains elements within its radical whose lengths are bounded above by a quadratic function of the number of states, then this condition is equivalent to the existence of reset words whose lengths are likewise bounded quadratically in the number of states. Consequently, establishing such an upper bound for radical words would settle the corresponding quadratic bound conjecture. Recall that the set of radical words forms a nilpotent ideal, so a suitable power of any such word is a reset word.

In this paper, we continue to develop the representation-theoretic approach to synchronizing automata, focusing specifically on the relationship between the lattice of ideals of the transition monoid of a given automaton $\A$ and its lattice of congruences $\text{Cong}(\A)$.
Congruences are equivalence relations on the state set $Q$ that are compatible with the automaton's action. The motivation for studying these congruences arises from a central theme of this paper, which addresses the following open problem:

\begin{opbm}[Hereditariness of the \v{C}ern\'y conjecture]
Let $\A$ be a synchronizing automaton, and $\sigma \in \mathrm{Cong}(\A)$ a non-trivial congruence. If the quotient automaton $\A/\sigma$ satisfies the \v{C}ern\'y conjecture, can this property be lifted to $\A$? In particular, does it guarantee that $\A$ admits a reset word of length within the \v{C}ern\'y bound?

\end{opbm}
If one can show that a \v{C}ern\'y reset word—a synchronizing word of length bounded by $(n-1)^2$—for the quotient automaton, always implies the existence of a corresponding \v{C}ern\'y reset word for the original automaton, then it would suffice to prove the conjecture for simple automata via a standard induction argument. 
In particular, we focus on the following main open problem:

\begin{conj}\label{conj: simple}
If the \v{C}ern\'y conjecture holds for all simple automata—i.e., automata whose only congruences are the identity and universal relations—then it holds for all automata in general.
\end{conj}

We were not able to prove this conjecture; however, our main theorem, stated in Theorem~\ref{corfinale}, provides the following reduction:
\begin{teo*}
If the \v{C}ern\'y conjecture holds for strongly connected simple, quasi-simple, and radical automata, then it holds in general. Moreover, every extremal automaton belongs to one of these classes.
\end{teo*}
Thus, the analysis of the \v{C}ern\'y conjecture can be restricted to three classes of automata: the \emph{simple}, the \emph{radical}, and the \emph{quasi-simple} ones. The latter two classes are described in detail in the sequel. It is worth noting that, when automata are viewed as unary algebras, these two classes are strictly contained in the class of subdirectly irreducible algebras; see \cite[Theorem~8.4]{Universal}. For a preliminary overview of all the classes discussed in this work, see Fig.~\ref{fig:inclusion-diagram}.
\\
The paper is organized as follows. 
We begin in Section~\ref{sec: standard approach} by tackling the problem with a ``brute-force'' approach, namely by applying the well-known \textit{Pin--Frankl} algorithm. It turns out that this technique is still insufficient to fully resolve the conjecture.
We then move on to study the \emph{lattice of congruences} of a given automaton. In Section~\ref{sec: lattice of congruences} we recall some basic facts about lattices of automaton congruences, with particular focus on the problem of calculating the atoms of such a lattice, which is crucial for computing the radical $\rad(\A)$ of a reset automaton $\A$. We also analyze in more detail the interplay between semisimplicity and the lattice of congruences, obtaining results that will later be used to establish the three-class reduction. 
In Section~\ref{sec: galois connection} we show a Galois connection between this lattice and the lattice of ideals of the transition monoid. This construction not only provides a recursive algorithm to compute the radical $\rad(\A)$ (developted in Section~\ref{sec: radical ideal and cong}), but, more importantly, offers a refined way to tackle the hereditariness problem. In particular, it is fundamental in proving our main results. The remainder of Section~\ref{sec: main result} is devoted to the study of these classes of automata. The quasi-simple case appears to be the more rigid of the two new classes, and we derive several structural results for it. The central result is Theorem~\ref{propidminsemisimp}, which shows that both simple and quasi-simple automata share the property that the transition monoid $\M(\A)$ admits a unique maximal ideal lying above $\Syn(\A)$. In contrast, for the radical case we were unable to obtain results of comparable depth. Most notably, Proposition~\ref{prop: structure radical} shows that in a radical automaton the ideal $\Syn(\A/\sigma)$, where $\sigma$ is the minimal non-trivial congruence of $\text{Cong}(\A)$ (also called monolith of $\A$), contains a unique non-trivial minimal ideal of $\M(\A)/\rad(\A^\star)$. Finally, Theorem~\ref{theo: bound on index} is a general result that bounds the index of nilpotency of $\rad(\A^\star)$ by the height of the congruence lattice $\text{Cong}(\A)$, revealing an interesting connection between the combinatorial structure of the automaton $\A$ and the algebraic structure of its transition monoid $\M(\A)$.

\section{Prerequisites}
For a set $A$, we denote $|A|$ the cardinality of $A$. The free monoid on $\Sigma$ is denoted by $\Sigma^{*}$. An ideal $I$ of a monoid $M$ is a subset of $M$ satisfying $MIM\subseteq I$.  All semigroup-theoretic notions used in this paper can be found in standard references on semigroup theory, such as \cite{Howie}.
In what follows, an \textit{automaton} $\A = (Q,\Sigma,\delta)$ refers to a deterministic finite automaton (DFA), where $Q$ and $\Sigma$ are finite sets denoting the set of \textit{states} and the \textit{alphabet}, respectively, and $\delta: Q \times \Sigma \to Q$ is the \textit{transition function}. We denote by $|\A| = |Q|$ the number of states of the automaton. Rather than the functional notation, we adopt an action notation, writing $q \cdot x$ in place of $\delta(q,x)$. This action extends naturally to the free monoid $\Sigma^*$, and then to any subset $F \subseteq Q$, by setting 
\[ 
F \cdot u = \{ q \cdot u \mid q \in u \}
\] 
for any $u \in \Sigma^*$. The \textit{transition monoid} of $\A$, denoted by $\M(\A)$, is the quotient of $\Sigma^*$ by the congruence $\equiv_{\A}$ defined by $u \equiv_{\A} v$ if and only if $q \cdot u = q \cdot v$ for all $q \in Q$ (i.e., the kernel of the action).
Henceforth, we will consider automata that are synchronizing. It is a well-known fact that the set 
\[
\mathrm{Syn}(\A) = \{ u \in \Sigma^* : |Q \cdot u| = 1 \}
\]
of the reset (or synchronizing) words of $\A$ is a two-sided ideal of the free monoid $\Sigma^*$. Furthermore, $\mathrm{Syn}(\A)$ is a regular language over $\Sigma$ (see \cite{Vo_Survey}). An element $u\in\Syn(\A)$ satisfying the \v{C}ern\'y bound, i.e. such that $|u|\leq (n-1)^2$ is said to be a \emph{\v{C}ern\'y reset word}. An automaton is said to be \textit{strongly connected} if its associated directed graph is strongly connected. We state here the following result, that will be fundamental in what follows:

\begin{prop}\label{prop: str conn}
    If the \v{C}ern\'y conjecture is solved for strongly connected synchronizing automata, then it holds in general.
\end{prop}
\begin{proof}
    See for instance \cite[Lemma~3.3, Proposition~3.3]{VolkSurv2}.
\end{proof}
Given an automaton $\A = (Q,\Sigma,\delta)$, we define a \emph{congruence} to be an equivalence relation $\sigma \subseteq Q \times Q$ which is compatible with the action, that is, for every $u \in \Sigma^*$ and $p,q \in Q$, if $p \,\sigma\, q$ then $(p \cdot u) \,\sigma\, (q \cdot u)$. The \emph{quotient automaton} is then defined as $\A/\sigma = (Q/\sigma, \Sigma, \delta_\sigma)$, where $\delta_\sigma$ is the induced transition function. It is immediate that the diagonal and universal relations on $Q$, namely $\Delta_\A$ and $\nabla_\A$, are congruences of $\A$. Let $\mathrm{Cong}(\A)$ denote the poset (w.r.t. the inclusion) of all congruences of $\A$.
For an element $u \in \M(\A)$, the kernel $\ker(u) \subseteq Q \times Q$ is the equivalence relation defined by  
\[
\ker(u) := \{(p,q) \in Q \times Q \mid p \cdot u = q \cdot u\}.
\] 
Note that, in general, this relation is not necessarily a congruence. Henceforth, we will need to consider kernels in quotient automata. We adopt the notation $\ker_{\mathcal{B}}$ to emphasize that the kernel relation is taken with respect to the automaton $\mathcal{B}$ and its state set. When no subscript is specified, the kernel is taken with respect to the 
ambient automaton under consideration, which, for most of this paper, will be denoted by $\A$.
\\
An important class of automata is given by the simple ones: an automaton $\A$ is said to be \emph{simple} if $\mathrm{Cong}(\A) = \{\Delta_\A, \nabla_\A\}$. A notable infinite class of simple automata is given by the \textit{\v{C}ern\'y automata} (see \cite{Vo_Survey} for their definition and, for instance, \cite{AlRo} for a proof of their simplicity).
We now briefly recall the algebraic notation used throughout the paper. Given an automaton $\A$, there exists an epimorphism $\pi: \Sigma^* \rightarrow \M(\A)$, where $\M(\A)$ is the transition monoid associated to $\A$. For convenience, we sometimes drop $\pi$ and identify ideals (and elements) in $\Sigma^*$ with ideals (and elements) in $\M(\A)$ and vice versa; so for instance if we are in the context of $\M(\A)$ we will write $\Syn(\A)$ instead of $\pi(\Syn(\A))$. It is well-known that $\M(\A)$ embeds as a monoid into the ring of matrices $\mathbb{M}_n(\mathbb{C})$, where $n$ is the number of states of $\A$ \cite{Steinb, BeBePe}. By slight abuse of notation, we will consider $\M(\A) \subseteq \mathbb{M}_n(\mathbb{C})$, so that $\pi: \Sigma^* \rightarrow \mathbb{M}_n(\mathbb{C})$ \cite{AlRo}. 
A $0$-monoid is a monoid equipped with a zero element $0$. In such a structure, any non-trivial minimal ideal $I \neq 0$ is called a $0$-minimal ideal. Throughout the paper we will often consider the $0$-monoid, defined as the Rees quotient $\A^\star := \M(\A)/\Syn(\A)$, and we denote by $\theta:\M(\A)\to \A^\star$ the corresponding Rees morphism. 
For a $0$-monoid $M$, a two-sided ideal $I \subseteq M$ is called \emph{nilpotent} with \emph{index of nilpotency} $m$ if $I^{m}=0$, where $m$ is the smallest positive integer with this property.
It is straightforward to verify that if two ideals $I$ and $J$ satisfy $I^{k}=J^{k}=0$, then $(I \cup J)^{2k}=0$.  
Hence, the union of nilpotent ideals is itself a nilpotent ideal.  
This observation allows us to define the largest nilpotent ideal $\rad(M)$ of $M$, known as the \emph{radical} of $M$. In the context of automata, the \emph{radical ideal} of a synchronizing automaton $\A$ is defined as $\rad(\A^\star)$ (see \cite{AlRo} for an alternative construction via the Wedderburn–Artin theorem).  
We will need to see this ideal in $\M(\A)$ and thus we set:
\[
\rad(\A) :=  \theta^{-1}(\rad(\A^\star)) \subseteq \M(\A).
\]  
By a slight abuse of notation, we write $\rad(\A)$ to refer to either $\rad(\A)$ or $\pi^{-1}(\rad(\A))\subseteq \Sigma^*$, depending on the context.
An automaton is called \emph{semisimple} if $\rad(\A^\star) = \{0\}$ (equivalently, if $\rad(\A) = \Syn(\A)$).
\\
Given an equivalence relation $\tau \subseteq Q \times Q$, we define  
\[
\text{Cong}_{\tau}(\A) := \{\sigma \in \text{Cong}(\A) \mid \sigma \subseteq \tau\}.
\]  
By \cite[Lemma 4]{AlRo}, for any $u \in \rad(\A) \setminus \Syn(\A)$ we have  
\[
\text{Cong}_{\ker(u)}(\A) \neq \{\Delta_\A\}.
\]  
This leads to the following result:
\begin{prop}
    If $\A$ is simple, then $\A$ is semisimple.
\end{prop}
Using the techniques developed later in this paper, we will provide an alternative proof of the above proposition.
Note that any simple automaton of at least three states is also strongly connected.    
An interesting question is under which conditions simplicity implies synchronizability.  
This problem has already been studied in the literature (see, for instance, \cite{Vol_Syn_Prim, Ryst_Prim, Araujo_Cameron}).  
Most remarkably, the following result provides a very general condition under which a simple automaton is synchronizing:

\begin{teo}\cite[Corollary 4]{Vol_Syn_Prim}
    Let $\A$ be a simple automaton with a letter of defect $1$.  
    Then, $\A$ is synchronizing. 
\end{teo}
Several algebraic conditions on an automaton guarantee its simplicity. For instance, in \cite{Ryst_Prim} it is shown that if the transition monoid  $\M(\A)/\Syn(\A)$ acts irreducibly, then the automaton is simple (called \emph{primitive} in the terminology of that paper).
A further condition implying the simplicity of an automaton is the existence of a subset of letters in the alphabet that generates a submonoid of $\M(\A)$ which is a group acting \emph{primitively} on the set of states. This topic has been extensively studied; see, for instance, \cite{Araujo_Cameron}.

We conclude this preliminary section by briefly recalling the theory of the Wedderburn--Artin decomposition of a synchronizing automaton 
(see \cite{AlRo} for a detailed treatment). 
For a synchronizing automaton $\A$, there exists a representation
\[
\varphi : \Sigma^*/\Syn(\A) \longrightarrow \mathbb{M}_{n-1}(\mathbb{C})
\]
such that $\varphi(\Sigma^*/\Syn(\A)) \cong \A^\star$, allowing us to view 
$\A^\star$ as a submonoid of $\mathbb{M}_{n-1}(\mathbb{C})$. 
We define $\mathcal{R}(\A)$ to be the subalgebra of $\mathbb{M}_{n-1}(\mathbb{C})$ 
generated by $\A^\star$; we refer to $\mathcal{R}(\A)$ as the 
\emph{synchronized $\mathbb{C}$-algebra associated with $\A$}.
If $\A$ is semisimple—which, as we will see, includes not only the simple case but also the quasi-simple case—then we can apply the 
Wedderburn--Artin theorem to $\mathcal{R}(\A)$, obtaining the decomposition
\[
\mathcal{R}(\A) \cong \mathbb{M}_{n_1}(\mathbb{C}) \times \cdots \times \mathbb{M}_{n_k}(\mathbb{C}).
\]
for some $k \geq 1$ and suitable positive integers $n_1, \ldots, n_k$. By composition, we obtain for each $i \in \{1, \ldots, k\}$ a morphism 
\[
\overline{\varphi}_i:\Sigma^* \longrightarrow \mathbb{M}_{n_i}(\mathbb{C}),
\]
and we define the $0$-semigroup $\mathcal{M}_i := \overline{\varphi}_i(\Sigma^*)$. 
Each $\mathcal{M}_i$ admits a unique $0$-minimal ideal $\mathcal{I}_i$, which is a $0$-simple semigroup. \\
In the non-semisimple case, let $\A$ be a non-semisimple automaton and let $\mathcal{R}(\A)$ denote its associated synchronized $\mathbb{C}$-algebra. Consider the monoid morphism $\rho:\Sigma^*\rightarrow\mathcal{R}(\A)$ obtained by composition. Let $\rad(\mathcal{R}(\A))$ be the \emph{Jacobson radical of $\mathcal{R}(\A)$} (see \cite{Lam}). Then 
\[
\rho^{-1}\!\big(\rad(\mathcal{R}(\A))\big)=\rad(\A).
\]
We define the semisimple quotient
\[
\overline{\mathcal{R}(\A)} := \mathcal{R}(\A)/\rad(\mathcal{R}(\A)),
\]
and, by the Wedderburn–Artin theorem, we may again consider the same decomposition as in the semisimple case. The following remark will be used later when discussing structural results for quasi-simple and radical automata.
\begin{oss}\label{rmk: connessione ideali e W.A.}
Let $\A$ be a synchronizing automaton and let $I \subseteq \M(\A)/\rad(\A)$ be a $0$-minimal ideal. Observe that
\[
\forall i \in \{1,\ldots,k\}, \quad 
I \subseteq \overline{\varphi}_i^{-1}(\mathcal{I}_i) \ \Rightarrow \ 
\overline{\varphi}_i(I) \subseteq \mathcal{I}_i.
\]
By minimality, we must have either $\overline{\varphi}_i(I)=\mathcal{I}_i$ or $\overline{\varphi}_i(I)=0$. Since $I \neq \rad(\A)$, it follows that $\overline{\varphi}_i(I)=\mathcal{I}_i$ for some $i \in \{1,\ldots,k\}$. This suggests that, in order to analyze the Wedderburn–Artin decomposition of a given automaton, it is natural to focus on the study of the $0$-minimal ideals of $\M(\A)/\rad(\A)$ (or $\A^\star$ in the semisimple case).
\end{oss}

\section{A standard approach to the hereditariness of the Černý conjecture}\label{sec: standard approach}
In this section, we analyze standard techniques used to address the hereditariness of the \v{C}ern\'y conjecture. We show that these methods face significant limitations and cannot be extended to a general proof, which motivates the need for a new approach. The challenges discussed here provide the basis for the more algebraic perspective developed in the following sections.

\begin{prop}\label{prop12}
    Let $\A$ be a strongly connected synchronizing DFA and $\sigma\in\text{Cong}(\A)$ a non-trivial congruence. Assume that $\sigma$ admits a 1-class or a 2-class: then if $\A/\sigma$ satisfies the Černý conjecture, we have that $\A$ satisfies the Černý conjecture as well. 
\end{prop}

\begin{proof}
    Let $u\in\Syn(\A/\sigma)$ be a Černý reset word. Let us consider the two following cases:
    \begin{itemize}
        \item assume that $\sigma$ admits a 1-class and let $p\in Q$ such that $[p]_\sigma = \{p\}$. We have that $Q\cdot u\subseteq [q]_\sigma$ for some $q\in Q$: by the strongly connectedness of $\A$ (and thus of $\A/\sigma$), we have that $\exists u_0\in\Sigma^*$ such that $[q]_\sigma \cdot u_0 = [p]_\sigma$ in $\A/\sigma$ with $|u_0|\leq n-2$. At this point we clearly get $Q\cdot uu_0 \subseteq [p]_\sigma = \{p\} \Rightarrow uu_0\in\Syn(\A)$ with:
        $$|uu_0| \leq (n-2)^2+n-2 = n^2- 4n +4 +n -2 = (n-1)^2 - n + 1 < (n-1)^2$$
        which concludes.
        \item assume that $\sigma$ admits a 2-class $[p]_\sigma$ for some $p\in Q$ and no 1-classes. By the same technique used in the previous point, we obtain a word $v\in\Syn(\A/\sigma)$ with $Q\cdot v \subseteq [p]_\sigma = \{p,q\}$. By means of the Pin-Frankl Algorithm (see \cite{Vo_Survey}) we have that $\exists v_0\in\Sigma^*$ such that $|\{p,q\}\cdot v_0 | = 1$ and:
        $$|v_0| \leq {n\choose 2} = \frac{n(n-1)}{2} = \frac{n^2-n}{2}$$
        and thus we have that $vv_0\in\Syn(\A)$ with:
        $$|vv_0| \leq (n/2-1)^2 + n/2 - 1 + \frac{n^2-n}{2} = \frac{3n^2}{4} - n \leq (n-1)^2, \mbox{ for }n\ge 2$$
        which covers also this case.
    \end{itemize}
    The two cases above complete the proof.
\end{proof}

The above case illustrates that we indeed encounter challenges when dealing with a “large” value of the cardinality of the smallest equivalence class of $\sigma$. In general, we have the following:

\begin{prop}\label{finqs}
    Let $\A$ be a synchronizing DFA, $\sigma\in\text{Cong}(\A)$ non-trivial and assume that the Černý conjecture holds for $\A/\sigma$. Let $m:=\min_i{|[p_i]_\sigma|}$ and assume $m\geq 3$: then $\exists u\in \Syn(\A)$ such that:
    $$|u|\leq \frac{n^2}{m^2}-\frac{n}{m} + \frac{n^3-n}{6} - \frac{(n-m)^3+3(n-m)^2+2(n-m)}{6}.$$
\end{prop}

\begin{proof}
    Let $u_0\in\Syn(\A/\sigma)$ be a Černý reset word for $\A/\sigma$ and $u_1\in\Sigma^*$ such that $Q\cdot u_0u_1\subseteq [p]_\sigma$ with $|[p]_\sigma| = m$. With the same technique adopted in the second case of the previous proof, again by the Pin-Frankl algorithm we have that $\exists v\in\Sigma^*$ such that $|[p]_\sigma\cdot v| = 1$ and:
    $$|v|\leq \sum_{k=2}^{m}{n-k+2\choose 2} = \frac{n^3-n}{6}-\frac{(n-m)^3+3(n-m)^2+2(n-m)}{6}$$
    and thus $u:=u_0u_1v\in\Syn(\A)$ with:
    $$|u|\leq \frac{n^2}{m^2}-\frac{n}{m} + \frac{n^3-n}{6} - \frac{(n-m)^3+3(n-m)^2+2(n-m))}{6}$$
    which concludes.
\end{proof}
Observe that, by the previous bound, the cases $m=3$ and $m=4$ with $n\geq 9$ already resolve the hereditariness problem. More generally, the most difficult instances arise when $m=\lfloor n/2\rfloor$, which appears to represent the worst case. In the case $m=n/2$, the previous bound yields 
\[
|u|\leq \frac{7n^3}{48} - \frac{n^2}{8} - \frac{n}{3}+2,
\]
which asymptotically improves Shitov's bound, since $7/48\sim 0.14584 < 0.1653 < \alpha$, where $\alpha$ is the coefficient of $n^3$ given in \cite{Shitov}. Note, however, that in general this bound does not improve upon Shitov's result; it does so only when $\A/\sigma$ admits a Černý-reset word for some $\sigma\in\text{Cong}(\A)$. However, if Conjecture~\ref{conj: simple} were true, then for any non-simple synchronizing automaton $\A$ one could consider a non-trivial maximal $\sigma\in\text{Cong}(\A)$ and apply the preceding result to obtain a genuine improvement of the bound in general.

\section{The lattice of congruences of an automaton}\label{sec: lattice of congruences}
This section is dedicated to the study of the lattice of congruences of a given DFA $\A = (Q, \Sigma, \delta)$. Using the language of universal algebra, we may view an automaton as a unary algebra on $Q$ with $\Sigma = \{a_1, \ldots, a_m\}$ unary operations, see for instance \cite{Universal}.  
Therefore, the poset of congruences $\Con(\A)$, ordered by inclusion, is a finite lattice where for every $\sigma, \tau \in \Con(\A)$ one has
\[
\sigma \wedge \tau := \sigma \cap \tau, 
\qquad 
\sigma \vee \tau := \bigcap \{ \rho \in \Con(\A) \mid \rho \supseteq \sigma, \tau \}.
\]
The maximum (top element) is given by the universal relation $1 := \nabla_\A = Q \times Q$,
and the minimum (bottom element) by the identity relation $0 := \Delta_\A = \{ (q,q) : q \in Q \}$.
In lattice-theoretic terms, a non-trivial congruence $\theta \in \Con(\A)$, $\theta \neq 0$, is called an \emph{atom} if it covers $0$, that is, if there is no congruence strictly between $0$ and $\theta$. Note that two distinct atoms $\sigma, \tau$ are orthogonal, i.e.,  $\sigma\cap\tau=0$. If a DFA admits a unique atom, such congruence is called \emph{monolith}. In other words, a monolith (if exists) is the non-trivial minimum of $\text{Cong}(\A)$. 

For any pair of states $p, q \in Q$, we define the principal congruence, or the
\emph{congruence generated by} $\{p,q\}$, denoted by $\langle\{p,q\}\rangle$, as follows:
\[
\langle\{p, q\}\rangle := 
\bigcap \left\{ \sigma \in \operatorname{Cong}(\mathcal A) \ \middle| \ (p, q) \in \sigma \right\}.
\]
Calculating this congruence is quite efficient, and it can be done in $O(m|\A|)$ time, see \cite[Theorem 4]{Freese}. The main strategy of this algorithm is based on the following fact.
\begin{lemma}\label{lem: equivalence closure}
With the above notation, $\langle\{p, q\}\rangle$ is equal to the equivalence closure $\langle S\rangle_{eq}$ of the set $S$ defined by:
\[
S = \{p, q\} \cdot \Sigma^* := \left\{ (t, s)\in Q\times Q \ \middle| \ \exists u \in \Sigma^* \text{ such that } \{p, q\} \cdot u = \{t, s\} \right\}.
\]
\end{lemma}
\begin{proof}
    It follows from \cite[Theorem~5.5]{Universal} applied to a unary algebra.
\end{proof}
The following lemma easily follows from the definition of atom of a lattice. 
\begin{lemma}\label{lemma122}
    Let $\text{Cong}(\A)$ seen as a lattice. Then, for every atom $\sigma\in\text{Cong}(\A)$ we have $\sigma = \langle \{p,q\}\rangle$ for some $p,q\in Q$ with $p\neq q$.
\end{lemma}
For a finite algebra $A$, the Demel-Demlov\'a-Koubek (DDK) algorithm computes the atoms of the congruence lattice of $A$ in $O(m|A|^{r+1})$ time, where $r$ is the maximum arity of the operations in the algebra $A$, and $m$ is the number of operations, see \cite[Theorem~4.3]{DDK85}. For a deterministic finite automaton $\A$, the transitions  are unary operations ($r = 1$), which specializes the complexity to  $O(m|Q|^2)$. For the sake of completeness, since it is also recalled in the computation of the radical ideal, we outline a similar algorithm here. Consider the following preorder on
\[
\binom{Q}{2} := \{ X \subseteq Q \mid |X| = 2 \},
\]
defined by $\{p,q\} \preceq \{s,t\}
\quad \Longleftrightarrow \quad
\exists\, u \in \Sigma^{*} \text{ such that } \{p,q\}\cdot u = \{s,t\}.$
Let $\sim$ denote the symmetric part of $\preceq$, that is, the equivalence relation on
$\binom{Q}{2}$ defined by $X \sim Y
\quad \Longleftrightarrow \quad X \preceq Y \text{ and } Y \preceq X $.
We may then consider the induced poset $\bigl(\binom{Q}{2}/\!\sim,\preceq\bigr)$,
where
\[
[\{p,q\}]_{\sim} \preceq [\{s,t\}]_{\sim}
\quad \text{if and only if} \quad
\{p,q\} \preceq \{s,t\}.
\]
It is straightforward to check that $\{p,q\} \sim \{s,t\}
\quad \Longleftrightarrow \quad
\{p,q\}\cdot \Sigma^{*} = \{s,t\}\cdot \Sigma^{*}.$
Hence, by Lemma~\ref{lem: equivalence closure}, if $\{p,q\} \sim \{s,t\}$ then $\langle \{p,q\} \rangle = \langle \{s,t\} \rangle$.
Therefore, the map
\[
L : \binom{Q}{2}/\!\sim \;\longrightarrow\; \Con(\mathcal A),
\qquad
L([\{p,q\}]_{\sim}) := \langle \{p,q\} \rangle,
\]
is a well-defined, surjective, antimorphism of posets.
Recall that the \emph{coatoms} of a lattice are its maximal proper elements.
It is straightforward to check that the preimage of an atom of
$\Con(\A)$ contains a coatom of $\binom{Q}{2}/_{\sim}$.
Hence, the set of atoms of $\Con(\A)$ is in bijection with a subset of the
coatoms of $\binom{Q}{2}/_{\sim}$. The poset $\binom{Q}{2}/_{\sim}$ admits a natural graph-theoretic interpretation. Consider the \emph{pairs digraph} $\mathcal{P}$ whose vertex set is $\binom{Q}{2}$, and where there is a directed edge $\{p,q\} \longrightarrow \{s,t\}$ whenever there exists a letter $a \in \Sigma$ such that
$\{p,q\} \cdot a = \{s,t\}$. The preorder $\preceq$ coincides with reachability in $\mathcal{P}$, while the equivalence relation $\sim$ corresponds to mutual reachability. Thus, each equivalence class $[\{p,q\}]_{\sim}$ is precisely a strongly
connected component (SCC) of $\mathcal{P}$ in the usual graph-theoretic sense. Consequently, the Hasse diagram of $\binom{Q}{2}/_{\sim}$ is isomorphic to the graph whose vertices are the SCC of $\mathcal{P}$ and where one SCC $[\{p,q\}]_{\sim}$ is below $[\{p',q'\}]_{\sim}$ if there is a directed edge connecting $[\{p,q\}]_{\sim}$ to $[\{p',q'\}]_{\sim}$. In this representation, the coatoms of $\binom{Q}{2}/_{\sim}$ are in one-to-one correspondence with the \emph{sink-SCC}, that is, SCCs of $\mathcal{P}$ that have no outgoing edges to other SCCs. This interpretation leads to the following meta-algorithm
$\textnormal{Atom}(\A)$ for computing the atoms of $\Con(\A)$.
The algorithm is a modification of Tarjan's classical procedure for
computing strongly connected components. The pairs digraph $\mathcal{P}$ is 
explored via a depth-first search. During the execution, the algorithm 
maintains a global variable $S \subseteq \binom{Q}{2}$ whose role is to 
record pairs $\{p,q\}$ such that the corresponding equivalence class 
$[\{p,q\}]_{\sim}$ is known \emph{not} to correspond to an atom, or for 
which there is another non-trivial pair $(u,v) \in \langle \{p,q\} \rangle$ 
not in $S$. Whenever the algorithm identifies a sink strongly connected component 
$C = [\{s,t\}]_{\sim}$, it performs a minimality check by computing the 
congruence $\langle \{s,t\} \rangle$ and testing whether every non-trivial 
pair of this congruence already belongs to $S$. If this is the case, 
then $\langle \{s,t\} \rangle$ is minimal among non-trivial congruences 
and hence defines an atom of $\text{Con}(\mathcal{A})$. Otherwise, there 
exists a pair $(u,v) \in \langle \{s,t\} \rangle$ such that $\{u,v\} \notin S$. 
In this case, one has $\langle \{u,v\} \rangle \subseteq \langle \{s,t\} \rangle$, 
so $\langle \{u,v\} \rangle$ is a candidate to be an atom. The algorithm 
then adds all pairs in $C$ to $S$ and recursively applies the same 
procedure to the SCC $[\{u,v\}]_{\sim}$. We emphasize that the main loop of the procedure is executed only for pairs $\{u,v\} \notin S$. Once a pair is inserted into $S$, it is never considered again as a candidate for generating an atom. This guarantees that each pair in $\binom{Q}{2}$ is processed at most once in the minimality checks. As a consequence, the overall cost of the algorithm is $O(m \cdot |Q|^2)$.
\\
We conclude this section regarding congruences of an automaton with the following construction, which will be useful in the sequel. Let $\A$ be a DFA, and let $\sigma \in \text{Cong}(\A)$. Given a congruence $\rho \in \text{Cong}(\A/\sigma)$, we define the \emph{lifting} of $\rho$ as the congruence $\overline{\rho} \in \text{Cong}(\A)$ given by:
\[
(p, q) \in \overline{\rho} \iff ([p]_\sigma, [q]_\sigma) \in \rho.
\]

Observe that $\overline{\rho}$ is indeed a congruence: since $\rho$ is an equivalence relation, its lifting remains an equivalence relation as well. To verify closure under transition, let $u \in \Sigma^*$ and suppose that $(p, q) \in \overline{\rho}$. Since $\rho$ is a congruence on $\A/\sigma$, we have:
\[
([p]_\sigma, [q]_\sigma) \cdot u = ([p]_\sigma \cdot u, [q]_\sigma \cdot u) = ([p \cdot u]_\sigma, [q \cdot u]_\sigma) \in \rho,
\]
which, by the definition of $\overline{\rho}$, implies that $(p \cdot u, q \cdot u) \in \overline{\rho}$. We can now state the following:

\begin{prop}\label{propliftcong}
Let $\sigma \in \text{Cong}(\A)$, then $\text{Cong}(\A/\sigma)$ embeds into $\text{Cong}(\A)$ as a sublattice.
\end{prop}

\begin{proof}
    By \cite[Theorem 6.20]{Universal} we have that $\text{Cong}(\A/\sigma) \cong [\sigma, \nabla_\A]$, where $[\sigma, \nabla_\A]$ is the sublattice of $\text{Cong}(\A)$  defined as follows:
    $$[\sigma, \nabla_\A] :=\{\tau \in \text{Cong}(\A) \ | \ \sigma \subseteq \tau\}.$$
    Here the isomorphism is given by $\rho\mapsto\overline{\rho}$.
\end{proof}

\subsection{Semisimplicity vs congruences}
In what follows, we give some sufficient conditions for semisimplicity that depend on the structure of the lattice $\text{Cong}(\A)$, we also state some results which will be used later. We conclude this subsection with Example~\ref{semisimpnonsimp}, which shows that the class of semisimple automata is strictly larger than the class of simple automata—an open problem that was, in some sense, implicitly raised in \cite{AlRo}. Let us begin with a basic fact:

\begin{lemma}\label{intcong}
Let $\mathcal{F} \subseteq \text{Cong}(\A)$ be nonempty. Then, for every $q \in Q$, we have:
\[
\bigcap_{\sigma \in \mathcal{F}} [q]_{\sigma} = [q]_{\bigcap \mathcal{F}}.
\]
\end{lemma}
\begin{proof}
Let $p$ be a state in $Q$. By definition, $p$ belongs to the intersection $\bigcap_{\sigma \in \mathcal{F}} [q]_\sigma$ if and only if $(p, q) \in \sigma$ for every $\sigma \in \mathcal{F}$. This means that $(p, q)$ belongs to the intersection $\bigcap \mathcal{F}$, and so $p \in [q]_{\bigcap \mathcal{F}}$.
Conversely, if $p \in [q]_{\bigcap \mathcal{F}}$, then $(p, q)$ belongs to every $\sigma \in \mathcal{F}$, which implies $p \in [q]_\sigma$ for each $\sigma$. Hence, $p$ belongs to the intersection on the left-hand side.
\end{proof}

\begin{lemma}\label{synquo}
    Let $\sigma\in\text{Cong}(\A)\setminus\{\Delta_\A,\nabla_\A\}$. Then if $\Syn(\A/\sigma)\subseteq \rad(\A)$, we have that $\rad(\A)=\rad(\A/\sigma)$.
\end{lemma}

\begin{proof}
    By the nilpotency of $\rad(\A/\sigma)$ in $\Sigma^*/\Syn(\A/\sigma)$, we have that $\exists k\geq 1$ such that: 
    $$\rad(\A/\sigma)^k\subseteq \Syn(\A/\sigma)\subseteq \rad(\A).$$
    But then, by the nilpotency of $\rad(\A)$ in $\Sigma^*/\Syn(\A)$ we have that $\exists m\geq 1$ such that $\rad(\A)^m\subseteq \Syn(\A)$, and thus:
    $$\rad(\A/\sigma)^{km}\subseteq \rad(\A)^m \subseteq \Syn(\A)$$
    and thus we have that $\rad(\A/\sigma)$ is a nilpotent ideal of $\Sigma^*/\Syn(\A)$. Being $\rad(\A)$ the biggest nilpotent ideal of $\Sigma^*/\Syn(\A)$, we must have $\rad(\A/\sigma) \subseteq\rad(\A)$, hence $\rad(\A)=\rad(\A/\sigma)$.
\end{proof}

\begin{prop}\label{intsin}
Let $\mathcal{F} \subseteq \text{Cong}(\A)$ with $\mathcal{F} \neq \emptyset$. Define $\rho := \bigcap \mathcal{F}$. Then:
\[
S := \bigcap_{\sigma \in \mathcal{F}} \Syn(\A/\sigma) = \Syn(\A/\rho),
\]
and in particular,
\[
R := \bigcap_{\sigma \in \mathcal{F}} \rad(\A/\sigma) = \rad(\A/\rho).
\]
\end{prop}

\begin{proof}
First observe that for any pair of congruences $\sigma, \tau \in \text{Cong}(\A)$, the inclusion $\sigma \subseteq \tau$ implies:
\[
\Syn(\A/\sigma) \subseteq \Syn(\A/\tau) \quad \text{and} \quad \rad(\A/\sigma) \subseteq \rad(\A/\tau).
\]
Applying this to the family $\mathcal{F}$, we immediately deduce that $\Syn(\A/\rho) \subseteq S$. Now, let $v \in S$. This means that $v \in \Syn(\A/\sigma)$ for every $\sigma \in \mathcal{F}$. Then, for each $\sigma$, the image of $Q$ under $v$ is contained in a single $\sigma$-class. In particular, there exists a state $p \in Q$ such that:
\[
Q \cdot v \subseteq \bigcap_{\sigma \in \mathcal{F}} [p]_\sigma = [p]_\rho,
\]
where the equality follows from Lemma~\ref{intcong}. This implies that $v \in \Syn(\A/\rho)$, and hence $S \subseteq \Syn(\A/\rho)$. Together with the previous inclusion, we conclude that $S = \Syn(\A/\rho)$.

For the second claim, we already know from the inclusion property above that $\rad(\A/\rho) \subseteq R$. Next, for every $\sigma \in \mathcal{F}$, there exists an integer $k_\sigma \geq 1$ such that $\rad(\A/\sigma)^{k_\sigma} \subseteq \Syn(\A/\sigma)$. Define:
\[
K := \prod_{\sigma \in \mathcal{F}} k_\sigma.
\]
Then, for every $\sigma \in \mathcal{F}$, we have:
\[
\rad(\A/\sigma)^K \subseteq \Syn(\A/\sigma).
\]
Since $R$ is contained in each $\rad(\A/\sigma)$, it follows that:
\[
R^K \subseteq \rad(\A/\sigma)^K \subseteq \Syn(\A/\sigma) \quad \text{for all } \sigma \in \mathcal{F},
\]
and hence,
\[
R^K \subseteq \bigcap_{\sigma \in \mathcal{F}} \Syn(\A/\sigma) = \Syn(\A/\rho).
\]
This shows that $R$ is a nilpotent ideal of the monoid $\Sigma^*/\Syn(\A/\rho)$, and therefore must coincide with $\rad(\A/\rho)$, as claimed.
\end{proof}
The following proposition shows that quotienting an automaton by the intersection of congruences that individually give rise to semisimple quotient automata also results in a semisimple automaton.

\begin{prop}\label{intsems}
    Let $\mathcal{F}\subseteq\text{Cong}(\A)$ be a non-empty collection of congruences such that $\A/\sigma$ is semisimple for every $\sigma\in\mathcal{F}$. Then, if $\rho:=\cap\mathcal{F}$ we have that $\A/\rho$ is semisimple.
\end{prop}

\begin{proof}
    Observe that:
    \[\begin{split}\Syn(\A/\rho) =\Syn(\A/\cap\mathcal{F})&=\bigcap_{\sigma\in\mathcal{F}}\Syn(\A/\sigma)\text{ by Proposition~\ref{intsin}} \\
    &= \bigcap_{\sigma\in\mathcal{F}}\rad(\A/\sigma) \ \text{ since each $\A/\sigma$ is semisimple }\\ 
    &= \rad(\A/\rho) \ \text{ by Proposition~\ref{intsin}} \end{split}\]
    hence $\A/\rho$ is semisimple.
\end{proof}
As a direct consequence of the previous result, we have the following interesting corollaries.
\begin{cor}\label{corintsem}
    If there are $\sigma_1,\sigma_2\in\text{Cong}(\A)$ with $\A/\sigma_1$, $\A/\sigma_2$ semisimple, and $\sigma_1\cap\sigma_2 =\Delta_\A$, then $\A$ is semisimple.
\end{cor}
For any maximal congruence $\sigma\in\text{Cong}(\A)\setminus\{\nabla_\A\}$, it is straightforward to see that $\A/\sigma$ is clearly simple and thus semisimple. This shows that for every synchronizing automaton $\A$ there is a congruence $\sigma\in\text{Cong}(\A)$ such that $\A/\sigma$ is semisimple. The following corollary strengthens this result:
\begin{cor}
    Let $\A$ be a synchronizing non-simple DFA, and let $\mathcal{M}$ be the set of its maximal non-trivial congruences. Then, if 
    $$
\sigma=\bigcap_{\tau\in \mathcal{M}}\tau
    $$
    we have that $\A/\sigma$ is semisimple.
\end{cor}
\begin{proof}
    Note that for a congruence $\tau\in\mathcal{M}$, $\A/\tau$ is simple, and thus semisimple. 
\end{proof}
Building on the previous results, we now present an example of an automaton that is semisimple but not simple.
\begin{ese}[Double Černý]\label{semisimpnonsimp}
    Let us now consider the following automaton $\A$:
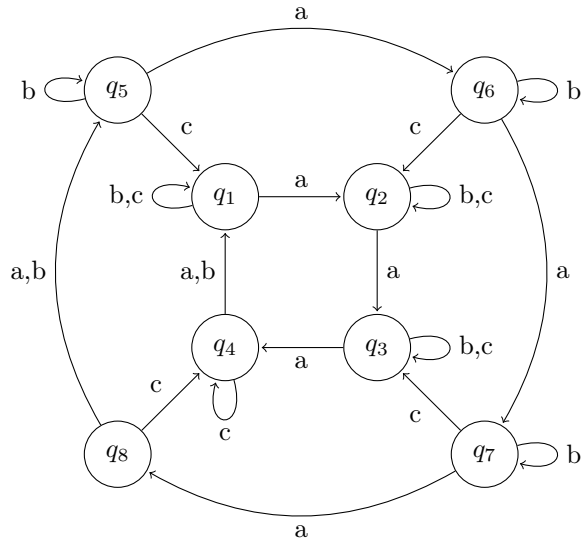
    \begin{center}\begin{tikzpicture}[shorten >=1pt,node distance=2cm,on grid,auto] 
        \node[state] (q_1)   {$q_1$}; 
        \node[state] (q_2) [right=of q_1] {$q_2$}; 
        \node[state] (q_3) [below=of q_2] {$q_3$}; 
        \node[state] (q_4) [below=of q_1] {$q_4$};
        \node[state] (q_5) [above left=of q_1]{$q_5$}; 
        \node[state] (q_6) [above right=of q_2] {$q_6$}; 
        \node[state] (q_7) [below right=of q_3] {$q_7$}; 
        \node[state] (q_8) [below left=of q_4] {$q_8$};
        \path[->] 
            (q_1) edge node {a} (q_2)
                  edge [loop left] node {b,c}  ()
            (q_2) edge node {a} (q_3)
                  edge [loop right] node {b,c} ()
            (q_3) edge node {a} (q_4)
                  edge [loop right] node {b,c} () 
            (q_4) edge node {a,b} (q_1)
                  edge [loop below] node {c} ()
            (q_5) edge [bend left] node {a} (q_6)
                  edge [loop left] node {b}  ()
                  edge node {c} (q_1)
            (q_6) edge [bend left] node {a} (q_7)
                  edge [loop right] node {b} ()
                  edge node [swap] {c} (q_2)
            (q_7) edge [bend left] node {a} (q_8)
                  edge [loop right] node {b} () 
                  edge node {c} (q_3)
            (q_8) edge [bend left] node {a,b} (q_5)
                  edge node {c} (q_4);
    \end{tikzpicture}\end{center}
    Observe that it is synchronizing: a reset word is given by $ba^3ba^3bc$. Consider the following congruences:
    $$\sigma_1:=\langle\{q_1,q_5
    \}\rangle, \ \sigma_2:=\langle\{q_5,q_6\}\rangle$$
    It is easy to show that $\A/\sigma_1 \cong\mathcal{C}_4$, where $\mathcal{C}_4$ stands for the Černý automaton of order $4$ which is simple (see \cite{Vo_Survey}), and $\A/\sigma_2$ is a 2-state automaton, and thus it is simple. Observe that $\sigma_1\cap\sigma_2 = \Delta_\A$, and by means of Corollary \ref{corintsem} we get that $\A$ is semisimple with:
    $$\Syn(\A) = \Syn(\A/\sigma_1) \cap \Syn(\A/\sigma_2).$$
\end{ese}

We conclude this section stating here this new conjecture which is a more general version of the Radical Conjecture stated in \cite{AlRo}: our aim in Section 7.1 will be to go as close as possible to its solution.

\begin{conj}[Semisimple conjecture]\label{cj: semisimple}
    If the Černý conjecture is solved for semisimple automata, then it holds in general.
\end{conj}

\section{Galois connection}\label{sec: galois connection}
In this section, we establish the existence of a Galois connection between the lattice \( \operatorname{Cong}(\A) \) of congruences on an automaton \( \A \) and the lattice \( (\IdM(\A), \subseteq) \) of ideals of the transition monoid \( \M(\A) \), where, without loss of generality, we sometimes consider such connections between \( \operatorname{Cong}(\A) \) and the lattice \( (\IdM(\A), \subseteq) \) of (two-sided) ideals of the free monoid \( \Sigma^* \). This is because it is often more convenient to view ideals as languages, that is, as subsets of \( \Sigma^* \), while working with ideals of \( \M(\A) \) is sometimes preferable due to the finiteness of the monoid, which ensures a finite ideal lattice. 

This connection provides a powerful tool for analyzing the interaction between ideals and congruences, and will be applied in the context of the Hereditariness problem. We begin with the following crucial lemmas.
\begin{lemma}
    Let $J\subseteq \text{M}(\A)$ be a two-sided ideal. Then $$\rho(J):=\bigcap_{u\in J}\text{ker}(u)$$
    is a congruence of $\A$.
\end{lemma}

\begin{proof}
    Fix $v\in\Sigma^*$ and observe that:
    \[\begin{split}
        p\,\rho(J)\,q &\Rightarrow p\cdot u = q\cdot u, \ \forall u\in J \\
        &\Rightarrow p\cdot vu = q\cdot vu, \ \forall u\in J \\
        &\Rightarrow (p\cdot v)\, \rho(J)\, (q\cdot v)
    \end{split}\]
    and this completes the proof.
\end{proof}
Observe that the proof of the previous lemma does not actually require \( J \) to be two-sided; it suffices that \( J \) is a left ideal of \( \text{M}(\A) \).

\begin{lemma}
    Let $\sigma\in\text{Cong}(\A)$. Then 
    $$I(\sigma):=\{u\in\text{M}(\A) \ | \ \text{ker}(u)\supseteq \sigma\}$$
    is a two-sided ideal of $\text{M}(\A)$.
\end{lemma}

\begin{proof}
    Let $u\in I(\sigma),\  x,y\in\text{M}(\A)$. We clearly have that:
    $$\sigma\subseteq\text{ker}(u) \iff \forall p\in Q, \ |[p]_\sigma\cdot u| = 1$$
    Now, note that we have the following inequalities
    \[
        |[p]_\sigma \cdot xuy|  \leq |[p\cdot x]_\sigma\cdot uy| 
         = |\{p\cdot xu\}\cdot y|
         = 1. \]
    Hence, we conclude that $\sigma\subseteq\text{ker}(xuy)$. 
\end{proof}

We can extend the above definition on the poset of two-sided ideals of $\Sigma^*$ with respect to the inclusion as follows. Let us define
$$\mathbb{I}(\sigma):=\{u\in\Sigma^* \ | \ \text{ker}(\pi(u)) \supseteq \sigma\}$$
where $\pi:\Sigma^*\rightarrow\text{M}(\A)$ is the usual epimorphism. 
\begin{oss}\label{congreg}
    Observe that $\mathbb{I}(\sigma) =\pi^{-1}(I(\sigma))$, hence it is a regular language being recognized by the morphism $\pi:\Sigma^*\rightarrow\text{M}(\A)$, see for instance \cite{Hopcroft}. We may explicitly calculate $\mathbb{I}(\sigma)$ without passing through the transition monoid via the following construction. Let $[p_1]_{\sigma}, \ldots, [p_m]_{\sigma}$ be the set of equivalence classes of $Q/{\sigma}$, and define
\[
\mathbb{I}_i(\sigma) := \{u \in \Sigma^* \mid |[p_i]_\sigma \cdot u| = 1\}.
\]
It is straightforward to verify that $\bigcap_{i=1}^m \mathbb{I}_i(\sigma) = \mathbb{I}(\sigma)$. Fix $i \in \{1, \ldots, m\}$, and let us show that $\mathbb{I}_i(\sigma)$ is accepted by the following automaton. Let $\mathbb{P}(\A) := (2^Q, \Sigma, \delta_p)$ be the power-automaton of $\A$, and define the automaton $\mathcal{B}_i := (2^Q, \Sigma, \delta_p, q_0, F)$, where $q_0 := [p_i]_\sigma$ and $F := \{\{p\} \mid p \in Q\}$. Observe that the language $L(\mathcal{B}_i)$ accepted by $\mathcal{B}_i$ is exactly $\mathbb{I}_i(\sigma)$. Therefore, $\mathbb{I}(\sigma) = \bigcap_i L(\mathcal{B}_i)$, which can be effectively computed by the standard construction involving the direct product of the automata $\mathcal{B}_i$.
\end{oss}

Let $\sigma \in \text{Cong}(\A)$. In what follows, with a slight abuse of notation, we will consider $\Syn(\A/\sigma)$ as a two-sided ideal of $\text{M}(\A)$, defined as
\[
\Syn(\A/\sigma) := \{x \in \text{M}(\A) \mid Q \cdot x \subseteq [p]_\sigma \ \text{for some } p \in Q\}.
\]
Let $\pi: \Sigma^* \rightarrow \text{M}(\A)$ be the usual monoid epimorphism. When clear from context, we will still use $\Syn(\A/\sigma)$ to refer to $\pi^{-1}(\Syn(\A/\sigma))$.

\begin{lemma}\label{lem: annh}
    For any $\sigma\in\text{Cong}(\A)$ we have
    $\Syn(\A/\sigma)\cdot I(\sigma) \subseteq \Syn(\A)$. Moreover, the projection of the ideal $I(\sigma)\cap \Syn(\A/\sigma)$ into $\A^\star$ is a nilpotent ideal of order at most two. 
\end{lemma}
\begin{proof}
    For every $u\in\Syn(\A/\sigma)$, $Q\cdot u\subseteq [p]_\sigma$ for some $p\in Q$. Let then $v \in I(\sigma)$: we have that
    $$|Q\cdot uv| \leq |[p]_\sigma \cdot v| = 1$$
    and thus $uv\in\Syn(\A)$. The second statement follows from the following inclusions:
    $$(I(\sigma)\cap \Syn(\A/\sigma))^2 \subseteq \Syn(\A/\sigma)\cdot I(\sigma) \subseteq \Syn(\A).$$
    % which shows that $I(\sigma)\cap \Syn(\A/\sigma)$ is a radical ideal of $\A^\star$ of order two. 
\end{proof}
Note that what we have just described can be lifted to the free monoid $\Sigma^*$ in case of $\Syn(\A/\sigma) \cdot \mathbb{I}(\sigma)$.
By means of the above lemmas, we are able to construct two poset antimorphisms, namely $\rho(-)$ and $I(-)$, between the posets $(\operatorname{Cong}(\A), \subseteq)$ and $(\IdM(\A), \subseteq)$, which together form a Galois connection. Let us first recall some basic facts about Galois connection theory. For a more detailed discussion, see \cite{GaloisC}.

\begin{defi}
    Let $(A,\le),(B,\le)$ be posets with $f:A\rightarrow B, \ g:B\rightarrow A$ antimorphisms. We say that the two antimorphisms $f,g$ form a \textit{Galois connection} if 
    $$\forall a\in A, b\in B.\ a\leq g(f(a)), \ b \leq f(g(b)).$$ 
\end{defi}

We have the following proposition.

\begin{prop}\label{morfpos}
    With the above notation, the two antimorphisms \[\rho:\IdM(\A)\rightarrow\text{Cong}(\A), \quad I:\text{Cong}(\A) \rightarrow \IdM(\A)\] form a Galois connection.
\end{prop}

\begin{proof}
    Let $J\in\IdM(\A)$ and $u\in J$, then $\text{ker}(u) \supseteq \rho(J)$ which implies $ u\in I(\rho(J))$. Hence, we have the inclusion $J\subseteq I(\rho(J))$. Let us now consider $\sigma \in \text{Cong}(\A)$. Observe that for any $u\in I(\sigma)$ we have $\text{ker}(u) \supseteq \sigma$, from which we deduce $\sigma\subseteq\rho(I(\sigma))$.
\end{proof}

Let us now state the following basic result concerning Galois connections:

\begin{prop}
    Let $f:A\rightarrow B, \ g:B \rightarrow A$ be a Galois connection. Then $\forall a\in A, \ b\in B$ we have that $f(a), \ g(b)$ are \textit{fixed points} for $f\circ g$ and $g\circ f$, respectively, i.e., 
    $$f(g(f(a))) = f(a), \ g(f(g(b))) = g(b).$$
\end{prop}

The above result can be rephrased in our context as follows: 
\begin{oss}
    Let $K\subseteq \Sigma^*$ be an ideal, $\tau\in\text{Cong}(\A)$, $\sigma:=\rho(K)$ and $J:=I(\tau)$. Then we have that:
    $$\rho(I(\sigma)) = \sigma, \ I(\rho(J)) = J.$$
\end{oss}

\begin{prop}\label{prop: annih quotient}
Let $\A$ be a semisimple DFA, and let $K \in\IdM(\A)$ and $\tau \in \text{Cong}(\A)$ be such that $I(\rho(K)) = K$ and $\rho(I(\tau)) = \tau$. Then, for any $u \in \Sigma^*$:
\[
\begin{split}
    u \cdot I(\tau) \subseteq \Syn(\A) & \Leftrightarrow u \in \Syn(\A/\tau), \\
    u \cdot K \subseteq \Syn(\A) & \Leftrightarrow u \in \Syn(\A/\rho(K)).
\end{split}
\]
\end{prop}
\begin{proof}
Assume $u \cdot I(\tau) \subseteq \Syn(\A)$ and suppose, for contradiction, that $u \notin \Syn(\A/\tau)$. Then there exist $p, q \in Q \cdot u$ such that $[p]_\tau \ne [q]_\tau$ and $p \cdot v = q \cdot v$ for every $v \in I(\tau)$. By definition of $\rho(I(\tau))$, this implies $(p,q) \in \rho(I(\tau)) = \tau$, which contradicts $[p]_\tau \ne [q]_\tau$. Therefore, $u \in \Syn(\A/\tau)$. The other implication is an immediate consequence of Lemma \ref{lem: annh}. \\
The second implication is proved similarly: suppose $u \cdot K \subseteq \Syn(\A)$ and $u \notin \Syn(\A/\rho(K))$. Then there exist $p, q \in Q \cdot u$ such that $[p]_{\rho(K)} \ne [q]_{\rho(K)}$ and $p \cdot v = q \cdot v$ for every $v \in K$. But then $(p,q) \in \rho(K)$, contradicting the assumption that $[p]_{\rho(K)} \ne [q]_{\rho(K)}$. Hence, $u \in \Syn(\A/\rho(K))$. The other implication follows easily combining Lemma \ref{lem: annh} and Proposition \ref{morfpos}. 
\end{proof}

In particular, the previous proposition shows that the left annihilators of $ I(\tau)$ and $K$ seen as ideals of $\A^\star$ correspond to the projection of $\Syn(\A/\tau)$ and $\Syn(\A/\rho(K))$, respectively, under the natural projection onto $\A^\star = \text{M}(\A) / \Syn(\A)$.

\section{The radical ideal and the radical congruence}\label{sec: radical ideal and cong}

This section is mainly devoted to the non-semisimple case, with special attention 
to the radical ideal $\rad(\A^\star)$, its computation, and the associated 
congruence, studied via the Galois correspondence introduced in the previous section. In particular, we establish a connection between the \emph{index of nilpotency}  of $\rad(\A^\star)$ and the height of the lattice $\text{Cong}(\A)$.

We start with the following remark: given $I\subseteq \A^\star$ a two-sided ideal, if $I=0$ we define $\rho(I):=\nabla_\A$. Otherwise, if $I\neq 0$ we define $\rho(I)  := \rho(\theta^{-1}(I))$  where $\theta:\text{M}(\A)\rightarrow\A^\star$ is the Rees morphism. Note that:

\[\rho(I)  = \rho(\theta^{-1}(I)) 
    = \rho((I\setminus\{0\})\cup\Syn
    (\A))
    =\bigcap_{u\in (I\setminus\{0\})\cup\Syn
    (\A)}\text{ker}(u) = \bigcap_{u\in I\setminus\{0\}} \text{ker}(u)
\]
which allows us to restrict the definition of $\IdM(\A)$ to ideals of $\A^\star$. From now on, by the above expression we will not make any difference between $\rho(I)$ and $\rho(\theta^{-1}(I))$ for a given ideal $I\subseteq\A^\star$.
\noindent We begin with the study of the congruence given by the radical ideal, denoted by $\rho:=\rho(\rad(\A))$. We will refer to such congruence as \textit{radical} congruence. Throghout the rest of this work, by abuse of notation we will use $\rad(\A)$ also to refer to $\rad(\A^\star)$ and for any $I\subseteq \Sigma^*$ two-sided ideal, we define $\rho(I):=\rho(\pi(I))$ where $\pi:\Sigma^*\rightarrow\text{M}(\A)$ is the usual epimorphism. \\
The next result is a generalization of \cite[Lemma 4]{AlRo}:

\begin{lemma}\label{lemrhocong}
    Let $\A$ be a DFA. Then, either $\A$ semisimple, in which case $\rho=\nabla_{\A}$, or $\rho\notin \{\nabla_{\A}, \Delta_{\A}\}$.
\end{lemma}

\begin{proof}
    If $\A$ is semisimple we clearly have that $\rho=\rho(\rad(\A)) = \rho(\Syn(\A))=\nabla_\A$.  
    Assume that $\A$ is not semisimple and so let $u\in\rad(\A)\setminus\Syn(\A)$. Then, there exist distinct states $ p,q\in Q$ such that $p\cdot u \neq q\cdot u$, hence $(p,q)\notin \rho$ which implies that $ \rho \neq \nabla_\A$. \\
    Let now $m\in\mathbb{N}$ be the nilpotency index of $\rad(\A^\star)$. By definition, we may find $u:=u_1\ldots u_{m-1} \in\rad(\A)\setminus\Syn(\A)$ such that $u_i\in\rad(\A)$ for every $i\in\{1,\ldots,m-1\}$, and $uv\in\Syn(\A)$ for every $v\in\rad(\A)$. Since $u\notin\Syn(\A)$, we have that $Q\cdot u\supseteq\{p,q\}$ with $p\neq q$, and thus $p\cdot v = q\cdot v$ for every $v\in\rad(\A)$. Therefore, $(p,q)\in\rho$, and so we conclude that $\rho\neq \Delta_\A$. 
\end{proof}

\begin{oss}
    The latter result immediately shows that simple $\Rightarrow$ semisimple.
\end{oss}

Recall that the \emph{height} of a (finite) lattice (or poset) is the maximum length of its chains. We will now prove a theorem connecting the height of the congruence lattice with the radical of the automaton, in particular by showing that the height 
bounds the \emph{index of nilpotency} of $\rad(\A)$. We first need to state the following lemma:

\begin{lemma}\label{lemma: SynRadSyn}
    Let $\A$ be a non-semisimple automaton, $\rho = \rho(\rad(\A))$ as above. We have that:
    $$\Syn(\A/\rho)\cdot \rad(\A) \subseteq \Syn(\A).$$
\end{lemma}

\begin{proof}
    It is enough to observe that, being $\rho:\IdM(\A)\rightarrow \text{Cong}(\A)$ an antimorphism of a Galois connection, we have that $\rad(\A) \subseteq I(\rho(\rad(\A))) = I(\rho)$, and thus by Proposition \ref{prop: annih quotient} we obtain:
    $$\Syn(\A/\rho) \cdot I(\rho) \subseteq \Syn(\A) \Rightarrow \Syn(\A/\rho) \cdot \rad(\A) \subseteq \Syn(\A)$$
    which concludes.
\end{proof}

\begin{teo}\label{theo: bound on index}
    Let $\A$ be a non-semisimple synchronizing DFA, and let 
    $m$ be the height of $\text{Cong}(\A)$. 
    Then
    \[
        \rad(\A)^{m-1} \subseteq \Syn(\A).
    \]
    In other words, the nilpotency index of $\rad(\A)$ is bounded above by the height of the congruence lattice minus one. This bound is tight.
\end{teo}

\begin{proof}
Let $\rho_0 := \Delta_\A$ and for each $ i\geq 0$ define iteratively the following congruences:
$$ \rho_{i+1} := \rho_{\A_i} =  \bigcap_{u\in\rad(\A_i)}\text{ker}_{\A_i}(u),\ \text{ where } \ \A_0 := \A \ \text{ and } \ \A_i:=\A_{i-1}/\rho_i \ \text{ for } \ i\geq 1$$
Observe that the lifting congruence $\overline{\rho}_i$ belongs to 
$\text{Cong}(\A)$ for every $i$, as defined in Proposition~\ref{propliftcong}. Furthermore, by iterating the same proposition we conclude that $\overline{\rho}_i\subseteq \overline{\rho}_{i+1}$ for every $i\geq 0$.
Since $\text{Cong}(\A)$ is finite, there exists an index $k$ such that 
$\A_k$ is semisimple. Indeed, assume by contradiction that $\A_i$ is not semisimple for every $i\geq 0$. By Lemma \ref{lemrhocong}, we have that $\rho_{i+1}$ is a non-trivial congruence of $\A_i$, and thus $\overline{\rho}_i\subsetneq\overline{\rho}_{i+1}$ for every $i\geq 0$. This gives us an infinite chain in $\text{Cong}(\A)$: 
$$\overline{\rho}_0\subsetneq \overline{\rho}_1 \subsetneq \overline{\rho}_2\subsetneq \ldots \subsetneq \overline{\rho}_i \subsetneq \ldots $$
contradicting the finiteness of $\text{Cong}(\A)$.
Let $k$ be such integer and observe that we have a chain in $\text{Cong}(\A)$ as follows
$$\overline{\rho}_0 = \Delta_\A \subsetneq \overline{\rho}_1 \subsetneq \ldots \subsetneq \overline{\rho}_k \subsetneq \overline{\rho}_{k+1}  = \nabla_\A$$
where again $\overline{\rho}_{k+1}  = \nabla_\A$ is a consequence of Lemma \ref{lemrhocong}. By Lemma~\ref{lemma: SynRadSyn} we have:
$$\Syn(\A_i)\rad(\A_{i-1}) \subseteq\Syn(\A_{i-1})$$
and by applying the above inclusion iteratively for any $i$ we obtain:
\[\begin{split}
    \Syn(\A_k)\rad(\A_{k-1}) &\subseteq\Syn(\A_{k-1})\\
    \Syn(\A_k)\rad(\A_{k-1})\rad(\A_{k-2})&\subseteq\Syn (\A_{k-2})\\
    &\dots \\
    \Syn(\A_k)\rad(\A_{k-1})\dots\rad(\A_0)&\subseteq\Syn(\A)\\
    \rad(\A)^{k+1}&\subseteq\Syn(\A)
\end{split}\]
where the last line is obtained by observing that $\rad(\A)\subseteq \rad(\A_i)$ for every $i$ and $\Syn(\A_k) = \rad(A_k)$. Thus we have a bound on the nilpotency index of $\rad(\A^\star)$.
By maximality of $m$, together with the fact that the $\overline{\rho}_1,\ldots,\overline{\rho}_k$ are non-trivial, we have $ k+2\le m$, and therefore
\[
\rad(\A)^{m-1} \subseteq \rad(\A)^{k+1} \subseteq \Syn(\A),
\]
which concludes the first part of the proof. \\
Let us now prove that the given bound on the nilpotency index is sharp for an infinite class of automata $\mathcal{S}_n$; we first present this class in the case $n=4$. Let $\mathcal{S}_4$ be the following DFA:
    
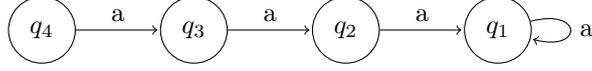
\begin{center}\begin{tikzpicture}[shorten >=1pt,node distance=2cm,on grid,auto] 
        \node[state] (q_1)   {$q_1$}; 
        \node[state] (q_2) [left=of q_1] {$q_2$};
        \node[state] (q_3) [left=of q_2] {$q_3$};
        \node[state] (q_4) [left=of q_3] {$q_4$}; 
        \path[->] 
            (q_1) edge [loop right] node {a}  ()
            (q_2) edge node {a} (q_1) 
            (q_3) edge node {a} (q_2)
            (q_4) edge node {a} (q_3);
    \end{tikzpicture}\end{center}
    It is easy to check that $\overline{\rho}_1=\ker(a)$ and $\overline{\rho}_2=\ker(a^2)$, with $\mathcal{S}_4/\rho_2$ being a simple automaton. It is also easy to check that the chain $\Delta_{\mathcal{S}_4} \subseteq \overline{\rho}_1 \subseteq \overline{\rho}_2 \subseteq \nabla_{\mathcal{S}_4}$ realizes the full poset $\text{Cong}(\A)$, from which we have $m-1 = 3$ and $\rad(\mathcal{S}_4)^3 = \Syn(\mathcal{S}_4)$ with $\rad(\mathcal{S}_4)^2\neq\Syn(\mathcal{S}_4)$. \\
    In general, one can generalize the construction of this automaton as follows: let $n\in\mathbb{N}$. We consider $Q=\{q_1,\ldots,q_{n}\}$, $\Sigma = \{a\}$ and:
    \[\delta(q_i,a)=\begin{cases}
        q_{i-1} & \text{ if } \ i\geq 2 \\
        q_1 & \text{ if } \ i= 1 \\
    \end{cases}\]
    also in this case it is easy to see that the chain is given by: 
    $$\overline{\rho}_1\subseteq \ldots\subseteq  \overline{\rho}_{n-1} \ \text{ with } \ \overline{\rho}_i=\ker(a^i)$$
    which gives a countable-class of examples for which the bound is sharp.
\end{proof}

\begin{oss}
In the setting of the previous theorem, observe that since $\A_k$ is semisimple, one might hope to apply this construction to prove Conjecture~\ref{cj: semisimple} by induction along the chain of congruences $\rho_i$, attempting to lift a Černý reset word from $\A_{i+1}$ to $\A_i$. 
Fix $i \in \{1,\ldots,k-1\}$ and assume that $\rho_{i+1}$ admits a $1$-, $2$-, or $3$-class as a congruence of $\A/\rho_i$. By Proposition~\ref{prop12}, we can easily construct a synchronizing word for $\A_i$ whenever one for $\A_{i+1}$ is given. Hence, in this case we would be able to solve the hereditariness problem.
Assume instead that, for every $i \in \{0,\ldots,k-1\}$, the congruence $\rho_{i+1}$ admits only classes of cardinality greater than $4$. In this situation, we were not able to solve the problem using any combinatorial technique to lift a Černý reset word. Nevertheless, one may observe that $|\A_i| \leq |\A_{i-1}|/4$ for every $i \in \{1,\ldots,k\}$, and thus $k+1 \leq \log_4(n)$, showing that the nilpotency index of $\rad(\A)$ is bounded by $\log_4(n)$.
\end{oss}

We conclude this section by describing our algorithm for the computation of the radical. We start by giving the following key fact: 

\begin{prop}\label{345}
    Let $\A$ be a synchronizing, non-semisimple DFA and $\rho:=\rho(\rad(\A))$ its associated radical congruence. Then, for any $\sigma\in\text{Cong}_\rho(\A)$ we have that:
    $$\rad(\A)=\rad(\A/\sigma)\cap I(\sigma).$$
    In particular, this holds for $\sigma=\rho$.
\end{prop}

\begin{proof}
    Clearly we have that $\rad(\A)\subseteq\rad(\A/\sigma)$. Being $\sigma\subseteq \rho$, we also have
    $$\rad(\A)\subseteq I(\rho)\subseteq I(\sigma)$$
    and thus $\rad(\A)\subseteq \rad(\A/\sigma)\cap I(\sigma)$. Observe now that, being $\rad(\A/\sigma)$ a nilpotent ideal of $(\A/\sigma)^\star$, there exists $k\geq 1$ such that $\rad(\A/\sigma)^k\subseteq \Syn(\A/\sigma)$ and thus:
    $$(\rad(\A/\sigma)\cap I(\sigma))^{k+1}\subseteq \rad(\A/\sigma)^k \cdot I(\sigma)\subseteq  \Syn(\A/\sigma) \cdot I(\sigma) \subseteq \Syn(\A) $$
    where $\Syn(\A/\sigma) \cdot I(\sigma) \subseteq \Syn(\A)$ is due to Lemma \ref{lem: annh}. This gives us that $\rad(\A/\sigma)\cap I(\sigma)$ is a nilpotent ideal of $\A^\star$. By the fact that $\rad(\A)$ is the biggest nilpotent ideal of $\A^\star$, we have $\rad(\A/\sigma)\cap I(\sigma) \subseteq \rad(\A)$ and thus $\rad(\A/\sigma)\cap I(\sigma) = \rad(\A)$.
    The last statement follows immediately by $\rho\in\text{Cong}_\rho(\A)$.
\end{proof}

Proposition~\ref{intsin} and the previous result are key ingredients for effectively computing the radical ideal of an automaton via Algorithm~\ref{Alg: radical}. The function \textit{Atom} used therein is the one defined in Section~\ref{sec: lattice of congruences}. The computation of $\mathbb{I}(\sigma)$ is also involved: by Remark~\ref{congreg}, it can be carried out via the power automaton. Moreover, the recursive procedure ensures that at each step, $\rad(\A/\sigma)$ is a regular language (since $\Syn(\A)$ is regular at step~\ref{linea: syn} of the algorithm), and thus the intersection of two regular languages can be computed using the standard product construction.
\begin{algorithm}[!h]
\caption{A recursive algorithm for calculating the radical for a synchronizing DFA $\A=(Q,\Sigma,\delta)$, knowing its atom congruences.}\label{Alg: radical}
\begin{algorithmic}[1]
    \Function{RadicalComputation}{$\A$}
        \State $\text{At} \gets $ \Call{Atom}{$\A$}\Comment{initializing the current atoms}
        \State $\rad(\A) \gets \Sigma^*$ \Comment{initializing the current radical}
        \If{At $\neq \{\nabla_\A\}$}      
            \If{$\text{At} = \{\sigma\}$}
                \State $\rad(\A/\sigma) \gets$ \Call{RadicalComputation}{$\A/\sigma$}
                \State $\rad(\A) \gets \rad(\A/\sigma)\cap \mathbb{I}(\sigma)$
            \Else
                \State consider $\sigma_1,\sigma_2\in\text{At}$ such that $\sigma_1\neq\sigma_2$
                \State $\rad(\A/\sigma_1) \gets$ \Call{RadicalComputation}{$\A/\sigma_1$}
                \State $\rad(\A/\sigma_2) \gets$ \Call{RadicalComputation}{$\A/\sigma_2$}
                \State $\rad(\A/\sigma)\gets \rad(\A/\sigma_1)\cap\rad(\A/\sigma_2)$
            \EndIf
        \Else
            \State $\rad(\A) \gets\Syn(\A)$   \label{linea: syn}
        \EndIf
        \State \Return $\rad(\A)$
    \EndFunction
\end{algorithmic}
\end{algorithm}

\section{Main result}\label{sec: main result}

This section is dedicated to proving the main results of the paper, namely Theorem \ref{corfinale} and the structural result Theorem~\ref{propidminsemisimp}. The approach we adopt is to treat separately the cases of semisimple and non-semisimple automata, proceeding by induction.

\subsection{The non-semisimple case}

The following definition singles out the subclass of non-semisimple automata that appears to be the most difficult to handle in the context of the Černý conjecture:

\begin{defi}
A non-semisimple, synchronizing DFA $\A$ is said to be \textit{radical} if 
\[
\rad(\A) \neq \rad(\A/\rho),
\]
where $\rho = \rho(\rad(\A))$, and where $\A$ admits a monolith (in other words,
$\text{Cong}(\A) \setminus \{\Delta_\A\}$ has a minimum).
\end{defi}
Recall that associated with $\rad(\A)$ there is the radical congruence 
$\rho = \rho(\rad(\A))$. Moreover, since $\A$ is radical, it admits a monolith $\sigma$. In this setting we obtain, by Proposition~\ref{345}, that
\[
\rad(\A) \;=\; I(\sigma) \cap \rad(\A/\sigma) 
\;=\; I(\rho) \cap \rad(\A/\rho).
\]
We now present two examples with $n=4$: the first is a non-semisimple, non-radical automaton, and the second is a radical automaton. These illustrate that both classes of automata are non-empty, and the construction can be easily generalized to any $n$.

\begin{ese}[Modified Černý automaton I]\label{eseMCA1}
Let $\A$ be the following DFA: 
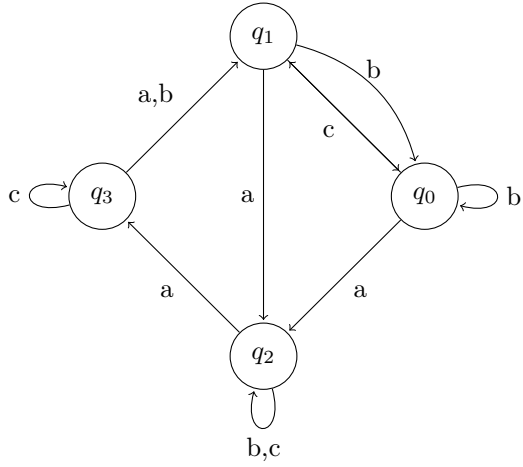
\begin{center}\begin{tikzpicture}[shorten >=1pt,node distance=3cm,on grid,auto] 
        \node[state] (q_0)   {$q_0$}; 
        \node[state] (q_1) [above left=of q_0] {$q_1$}; 
        \node[state] (q_2) [below left=of q_0] {$q_2$}; 
        \node[state] (q_3) [above left=of q_2] {$q_3$};
        \path[->] 
            (q_0) edge node {a} (q_2)
                  edge [loop right] node {b}  ()
                  edge node {c} (q_1)
            (q_1) edge [swap ]node {a} (q_2)
                  edge[bend left, above] node {b} (q_0)
                  edge node {} (q_0)
            (q_2) edge node {a} (q_3)
                  edge [loop below] node {b,c} () 
            (q_3) edge node {a,b} (q_1)
                  edge [loop left] node {c} ();
    \end{tikzpicture}\end{center}
    one may check that $ba^2b$ is a non-reset radical word, and that the only non-trivial congruence $\sigma$ is generated by the following partition $\big\{\{q_0,q_1\}, \{q_2\}, \{q_3\}\big\}$, with $\A/\sigma \cong \mathcal{C}_3$, the Černý automaton of order $3$. 
Thus, $\A$ is not semisimple and $\rho(\rad(\A)) = \sigma$. Let us now show that $\A$ is not radical, by showing that $\rad(\A/\sigma)\subseteq \rad(\A)$. Now let $u \in \Syn(\A/\sigma)\setminus\Syn(\A)$. 
We must then have $Q \cdot u = \{q_0,q_1\}$. 
Observe that for any $v \in \Sigma^*$,
\[
q_0 \cdot v \neq q_1 \cdot v 
\;\;\iff\;\;
v = c^k \quad \text{for some } k \in \mathbb{N}.
\]
But in this case $v \notin \Syn(\A/\sigma)$. 
Therefore, we must have $q_0 \cdot u = q_1 \cdot u$, which implies $u \in I(\sigma)$. This shows that $\Syn(\A/\sigma) \subseteq I(\sigma)$, and from Proposition~\ref{345} we obtain
\[
\rad(\A) \;=\; I(\sigma) \cap \Syn(\A/\sigma) 
\;=\; \Syn(\A/\sigma)\;=\; \rad(\A/\sigma).
\]
This gives a full description of the radical of $\A$ and shows that 
$\A$ cannot be radical.
\end{ese}

\begin{ese}[Modified Černý automaton II] 
We now slightly modify Example~\ref{eseMCA1}.  
Let $\A$ be the DFA:  

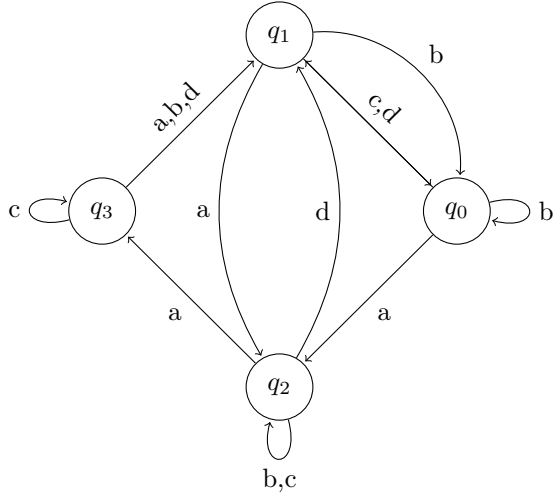
\begin{center}
\begin{tikzpicture}[shorten >=1pt,node distance=3.3cm,on grid,auto] 
    \node[state] (q_0)   {$q_0$}; 
    \node[state] (q_1) [above left=of q_0] {$q_1$}; 
    \node[state] (q_2) [below left=of q_0] {$q_2$}; 
    \node[state] (q_3) [above left=of q_2] {$q_3$};
    \path[->] 
        (q_0) edge node {a} (q_2)
              edge [loop right] node {b}  ()
              edge node {} (q_1)
        (q_1) edge [swap, bend right] node {a} (q_2)
              edge [bend left = 50] node {b} (q_0)
              edge [sloped] node {c,d} (q_0)
        (q_2) edge node {a} (q_3)
              edge [loop below] node {b,c} () 
              edge [bend right] node {d} (q_1)
        (q_3) edge [sloped] node {a,b,d} (q_1)
              edge [loop left] node {c} ();
\end{tikzpicture}
\end{center}
This automaton is the same as in the previous example, except that we have added a letter $d \in \Sigma$. Here, the only non-trivial congruence is again $\sigma = \rho(\rad(\A))$ defined as before. Observe that $Q \cdot d = \{q_0,q_1\} = [q_0]_\sigma$, and hence $d \in \Syn(\A/\sigma)$.  
However, $d \notin I(\sigma)$, since $q_0 \cdot d = q_1 \neq q_0 = q_1 \cdot d$.
It follows that 
\[
\rad(\A/\sigma) = \Syn(\A/\sigma) \not\subseteq I(\sigma),
\]
which implies 
\[
\rad(\A) \neq \Syn(\A/\sigma).
\]
Therefore, the automaton $\A$ must be radical.
\end{ese}

The following theorem points toward the validity of the semisimple 
Conjecture~\ref{cj: semisimple} and represents our first step toward 
the proof of the main result.

\begin{teo}\label{teosemrad}
    If the Černý conjecture holds for strongly connected semisimple and strongly connected radical automata, then it holds in general for strongly connected automata.
\end{teo}

\begin{proof}
    We will prove this statement by induction on $n=|Q|$. If $n=2$, the considered automaton is simple and thus semisimple so the statement holds. Let then $\A$ be a non-semisimple nor radical DFA, with $|Q|>2$. By means of Proposition \ref{prop: str conn} we can assume $\A$ to be strongly connected. Consider now the following cases:
    \begin{itemize}
        \item Assume that $\A$ has a monolith. Then by hypothesis, since $\A$ is not radical, we must have that $\rad(\A/\rho) = \rad(\A)$. By the induction hypothesis, let $u\in\Syn(\A/\rho)$ be a Černý reset word. If $\rho$ admits a 1- or 2-class, we are done by means of Proposition \ref{prop12} by strongly connectedness. Otherwise, by Lemma \ref{lemma: SynRadSyn} we have that:
        $$u\in\Syn(\A/\rho)\subseteq\rad(\A/\rho)=\rad(\A) \Rightarrow u^2 \in \Syn(\A/\rho)\cdot \rad(\A) \subseteq \Syn(\A)$$
        which implies $u^2 \in\Syn(\A)$, and since $|u| \leq (n/3-1)^2< (n/2-1)^2$ (since $\rho$'s smallest class has at least 3 elements), we have that $u^2$ strictly satisfies the Černý conjecture.
        \item Assume that $\text{Cong}(\A)$ admits two non-trivial minimal congruences $\sigma_1, \sigma_2$. 
        If one of them admits a $1$-class (say $\sigma_1$, without loss of generality), the claim follows by Proposition \ref{prop12}. 
        Otherwise, by Proposition~\ref{intsin} we have 
        \[
            \Syn(\A/\sigma_1) \cap \Syn(\A/\sigma_2) = \Syn(\A).
        \]
        Since each $\sigma_i$-class has at least two elements, by induction we obtain words $u_i \in \Syn(\A/\sigma_i)$ with $|u_i| < (n/2 - 1)^2$ for $i \in \{1,2\}$. 
        Thus, the concatenation $u_1u_2 \in \Syn(\A)$ is a reset word strictly satisfying the Černý conjecture.
    \end{itemize}
\end{proof}
We conclude this subsection by presenting a partial structural result concerning 
the $0$-minimal ideals of $\mathrm{M}(\A) / \rad(\A)$ in the radical case. 
Our motivation for focusing on these ideals stems from their connection with the 
Wedderburn--Artin decomposition of the automaton, as highlighted in 
Remark~\ref{rmk: connessione ideali e W.A.}. We need first the following lemma. Let $\xi : \A^\star \to \A^\star/\rad(\A) \cong 
\mathrm{M}(\A)/\rad(\A)$ denote the Rees morphism. 
For an ideal $I \subseteq \mathrm{M}(\A)/\rad(\A)$, 
we define 
\[
\rho(I) := \rho\!\bigl(\xi^{-1}(I)\bigr).
\] 

\begin{lemma}\label{lemminrad}
    Let $\A$ be a non-semisimple synchronizing automaton, $I_1,I_2 \subseteq \text{M}(\A)/\rad(\A)$  be two minimal non-trivial ideals such that $I_1\neq I_2$. Then we have that $\rho(I_1)$ or $\rho(I_2)$ has to be non-trivial.
\end{lemma}

\begin{proof}
    Observe that $I_1I_2 = 0$ can be rephrased by abuse of notation as $I_1I_2\subseteq \rad(\A)$ as ideals of $\A^\star$. Then, by another abuse of notation we can rephrase the latter inclusion as:
    $$(I_1I_2)^k = (I_1I_2\ldots I_1)I_2 \subseteq \Syn(\A) \ \text{ for some } \ k\geq 2$$
    by considering $I_1,I_2$ as ideals of $\text{M}(\A)$. Thus by taking $k$ minimal we must have either $(I_1I_2)^{k-1}I_1 \subseteq \Syn(\A)$ or $\rho(I_2)$ is non-trivial. Indeed, if $(I_1I_2)^{k-1}I_1 \not\subseteq \Syn(\A)$, let $u\in (I_1I_2)^{k-1}I_1\setminus\Syn(\A)$, and $p,q\in Q\cdot u$ with $p\neq q$. We must have $p\cdot v = q\cdot v$ for every $v\in I_2$, and thus $(p,q)\in\rho(I_2)$ which implies that $\rho(I_2)\neq\Delta_\A$. \\
    Assume then that $(I_1I_2)^{k-1}I_1 \subseteq \Syn(\A)$: by minimality of $k$ we have $(I_1I_2)^{k-1} \not\subseteq \Syn(\A)$, and thus in this case we must have that $\rho(I_1)$ is non-trivial, and this concludes the proof.
\end{proof}

This brings us to the following proposition:

\begin{prop}\label{prop: structure radical}
    Let $\A$ be a radical DFA, $\sigma\in\text{Cong}(\A)$ its monolith. Then $\Syn(\A/\sigma)$ contains exactly one 0-minimal ideal seen as an ideal of $\text{M}(\A)/\rad(\A)$.
\end{prop}

\begin{proof}
    First, observe that $\Syn(\A/\sigma)/\rad(\A)$ has to be non-trivial, for if otherwise we would have that $\Syn(\A/\sigma)\subseteq \rad(\A)$ as ideals in $\M(\A)$ and by Lemma \ref{synquo} this implies $\rad(\A) = \rad(\A/\sigma)$, against the hypothesis of $\A$ being radical. 
    Let $I_1,I_2 \subseteq \Syn(\A/\sigma)/\rad(\A)$ be two non-trivial minimal ideals. By means of the previous lemma we have that either $\rho(I_1)$ or $\rho(I_2)$ has to be non-trivial, without loss of generality assume $\rho(I_1)$ non-trivial. Then, by minimality we have $\sigma\subseteq \rho(I_1)$, from which we deduce 
    $$I_1 \subseteq I(\rho(I_1))/\rad(\A)\subseteq I(\sigma)/\rad(\A)$$
    by means of the usual properties of the Galois connection. Observe now that:
    $$I_1\subseteq \Syn(\A/\sigma)/\rad(A), \ I_1\subseteq I(\sigma) \Rightarrow I_1\subseteq (\Syn(\A/\sigma) \cap I(\sigma))/\rad(\A) = 0$$
    by means of Proposition \ref{345}, which is a contradiction. 
\end{proof}

\subsection{The semisimple case}

We now address the hereditariness problem for semisimple automata. We begin with the study of the relationship between congruences and ideals for semisimple but non-simple automata. In particular, we relate the minimal ideals of $\A^\star$ to the minimal non-trivial elements of $\mathrm{Cong}(\A)$ via our Galois connection.

\begin{lemma}\label{lem2congmin}
    Let $\A$ be a synchronizing DFA and $I_1,I_2\subseteq\A^\star$ be two distinct 0-minimal ideals. Then $\rho(I_1),\rho(I_2)$ are non-trivial congruences of $\A$. In particular, if $\A$ is simple then $\A^\star$ admits a unique 0-minimal ideal.
\end{lemma}

\begin{proof}
    Observe that, by the minimality of $I_1$ and $I_2$, we have that $I_1I_2 = I_2I_1 = 0$. Therefore, $\theta^{-1}(I_1)\theta^{-1}(I_2) = \theta^{-1}(I_2)\theta^{-1}(I_1) \subseteq \Syn(\A)$, where $\theta$ is the Rees morphism as previously defined. Let then $u\in I_1\setminus\Syn(\A)$ and $p,q\in Q\cdot u$ such that $p\neq q$. Observe that, for every $v\in I_2$ we must have $p\cdot v = q\cdot v$ and thus $(p,q)\in\rho(I_2)$ from which we deduce $\rho(I_2)\neq \Delta_{\A}$. Being $I_2 \neq 0$, we must also have that $\rho(I_2)\neq\nabla_{\A}$: this shows that $\rho(I_2)$ has to be non-trivial. Switching the indices and applying the same argument, one may show that $\rho(I_1)$ has to be non-trivial.
\end{proof}

\begin{prop}\label{congminsem}
    Let $\A$ be a semisimple DFA  admitting a monolith $\sigma$, and $I(\sigma)\neq \Syn(\A)$. Then $\Syn(\A)=\Syn(\A/\sigma)$. In particular, $\A/\sigma$ is semisimple.
\end{prop}

\begin{proof}
   To reach a contradiction assume that $\Syn(\A/\sigma)\neq\Syn(\A)$. Since $\Syn(\A/\sigma), I(\sigma)\neq \Syn(\A)$, in $\A^\star$ there are minimal non-trivial ideals $I_1, I_2$ with $I_1\subseteq \Syn(\A/\sigma), \ I_2\subseteq I(\sigma)$. With abuse of notation, we consider $I_1,I_2$ as ideals of $\text{M}(\A)$ by means of the Rees morphism $\theta:\M(\A)\rightarrow\A^\star$, and using $I_i$ also to refer to $\theta^{-1}(I_i)$. Assume that $I_1=I_2$: by Lemma~\ref{lem: annh} we would have that 
   $$I_1 \subseteq \Syn(\A/\sigma)\cap I(\sigma) \subseteq\rad(\A^\star) = \Syn(\A)$$
   which is clearly a contradiction. Assume then $I_1\neq I_2$ and observe that from Lemma~\ref{lem2congmin} together with the fact that $\sigma$ is monolith, we get $\rho(I_1) \supseteq\sigma$. Thus, by means of the Galois Connection we have  $I_1\subseteq I(\rho(I_1)) \subseteq I(\sigma).$ But then $I_1\subseteq \Syn(\A/\sigma)\cap I(\sigma)$, hence by Lemma~\ref{lem: annh} together with the semisimplicity of $\A$ we get $I_1 \subseteq  \Syn(\A)$, which is a contradiction. Hence, we may conclude that $\Syn(\A/\sigma)=\Syn(\A)$ holds.\\
   Now, the fact that $\A/\sigma$ is semisimple comes directly from Lemma \ref{synquo}, indeed we have:
   $$\rad(\A/\sigma) =  \rad(\A) = \Syn(\A) = \Syn(\A/\sigma)$$
   and this concludes the proof.
\end{proof}
In view of the previous proposition we give the following:
\begin{defi}
    A synchronizing DFA $\A$ is said to be \textit{quasi-simple} if $\A$ admits a monolith $\sigma$ and $I(\sigma) = \Syn(\A)$.
\end{defi}

The following proposition shows that the class of quasi-simple DFA is a subclass of the semisimple one:
\begin{prop}
    Let $\A$ be a quasi-simple DFA. Then $\A$ is semisimple.
\end{prop}

\begin{proof}
    Let $\sigma$ be the monolith of  $\A$ and $\rho:=\rho(\rad(\A))$. Assume $\A$ non-semisimple, by Lemma ~\ref{lemrhocong}, $\rho$ is non-trivial and thus $\sigma\subseteq \rho$. From this and the property of the Galois connection, we obtain:
    $$\rad(\A)\subseteq I(\rho)\subseteq I(\sigma) = \Syn(\A)$$
    where the last equality follows by the definition of quasi-simplicity. This shows that $\rad(\A)=\Syn(\A)$ and thus that $\A$ is semisimple.
\end{proof}

In general, we lack an efficient algorithm for producing ``short'' words $u \in I(\sigma)$. 
In the quasi-simple case, such a construction would be particularly useful in view of addressing 
the \v{C}ern\'y conjecture for this class. \\
The following example shows a quasi-simple synchronizing automaton. 
\begin{ese}
    Let $\A$ be the following DFA:
    
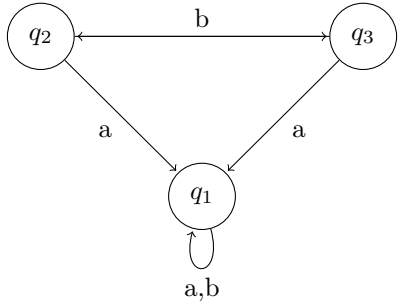
\begin{center}\begin{tikzpicture}[shorten >=1pt,node distance=3cm,on grid,auto] 
        \node[state] (q_1)   {$q_1$}; 
        \node[state] (q_2) [above left=of q_1] {$q_2$}; 
        \node[state] (q_3) [above right=of q_1] {$q_3$}; 
        \path[->] 
            (q_1) edge [loop below] node {a,b} (q_1)
            (q_2) edge node {b} (q_3)
                  edge [swap] node {a} (q_1) 
            (q_3) edge node {} (q_2)
                  edge node {a} (q_1);
    \end{tikzpicture}\end{center}
    Observe that the only non-trivial congruence is $\sigma:=\{\{q_1\},\{q_2,q_3\}\}$ and $\A$ is semisimple with $I(\sigma) = \text{M}(\A)\pi(a)\text{M}(\A) = \Syn(\A)$.
\end{ese}
We were not able to find a more complex example of a quasi-simple automaton, and we note that the known example is not even strongly connected. This naturally raises the question of whether such automata exist at all. In general, constructing quasi-simple automata appears to be a challenging task, 
and it would be interesting to discover at least an infinite family of such automata.  This difficulty suggests that the class of quasi-simple automata might in fact be very small. 
\\
One of the main difficulties we encountered in addressing the semisimple case is that, in general, it is not known  whether the quotient of a semisimple automaton remains semisimple.  In light of our reduction result, we attempted to resolve this question in the quasi-simple case; 
however, a complete resolution eluded us. This motivates the following open problem:
\begin{opbm}\label{prob: quotient of semisimple}
Let $\A$ be a quasi-simple DFA and $\sigma\in\text{Cong}(\A)$ be its monolith. Is it true that $\A/\sigma$ is semisimple?
\end{opbm}
Returning to quasi-simple automata and their peculiar structure, 
the following remark provides further evidence of the rigidity of their behavior with respect to the Galois connection.

\begin{oss}
    Let $\A$ be a quasi-simple automaton, and let 
    $\sigma \in \text{Cong}(\A)$ denote its unique non-trivial minimal congruence. 
    Then, for every non-trivial $\tau \in \text{Cong}(\A)$, we have $I(\tau)=\Syn(\A)$. 
    Indeed, by the Galois connection we have:
    \[
        \sigma \subseteq \tau 
        \ \Rightarrow \ 
        \Syn(\A) \subseteq I(\tau) \subseteq I(\sigma) = \Syn(\A),
    \]
    which forces $I(\tau)=\Syn(\A)$. 
    Furthermore, if $I \neq \Syn(\A)$ is an ideal of $\M(\A)$, then $\rho(I)=\Delta_\A$. 
    Together, these facts show that the only fixed points of the Galois connection are 
    $\Syn(\A)$ and $\M(\A)$ in $\mathcal{I}(\A)$, 
    and $\Delta_\A$ and $\nabla_\A$ in $\text{Cong}(\A)$.
\end{oss}

This observation highlights the highly constrained nature of quasi-simple automata. Such rigidity strongly suggests a close parallel with the simple case, a connection that is also motivated by the following theorem.

\begin{teo}\label{propidminsemisimp}
    Let $\A$ be a simple or quasi-simple DFA. Then $\A^\star$ admits a unique 0-minimal ideal.
\end{teo}

\begin{proof}
    If $\A$ is simple, we are done by Lemma~\ref{lem2congmin}. Assume then $\A$ to be quasi-simple. Let $I_1,I_2\subseteq\A^\star$ be two minimal non-zero ideals and assume $I_1\neq I_2$. Again by Lemma~\ref{lem2congmin}, we have $\rho(I_2)\neq \Delta_\A$: again if $\sigma$ is the monolith we get $I_2\subseteq I(\sigma) = 0$ and thus we can conclude that $I_2 = 0$, which is a contradiction.
\end{proof}

By some computation, it is not difficult to check that, for the class of \v{C}ern\'y automata $\mathcal{C}_n$, we have $\mathcal{R}(\mathcal{C}_n)\cong\mathbb{M}_{n-1}(\mathbb{C})$.  However, a complete result for the full class of simple automata is missing. This fact together with the above theorem and Remark~\ref{rmk: connessione ideali e W.A.} leads us to state the following conjecture

\begin{conj}
    Let $\A$ be either a simple or quasi-simple automata, $\mathcal{R}(\A)$ its associated synchronized $\mathbb{C}$-algebra. Then $\mathcal{R}(\A)\cong\mathbb{M}_{m}(\mathbb{C})$ for some $m\leq n-1$.
\end{conj}

We move now to address the hereditariness problem in the semisimple case, and we start by considering situations where more than one non-trivial congruence arises. We begin with the following lemma:

\begin{lemma}\label{335}
    Let $\A$ be a semisimple DFA which is not simple, and let $\mathcal{C} = \{\sigma_1,\ldots,\sigma_s\}\subseteq \text{Cong}(\A)$ be the set of its non-trivial minimal congruences. Assume that $s\geq 2$ and $\Syn(\A/\sigma)\neq \Syn(\A)$ for every $\sigma\in\mathcal{C}$. Then, there exist distinct 0-minimal ideals $I_1, \ldots, I_s$ in $\A^\star$ such that for every $i\in [1,s]$ we have:
    \[
        I_i \subseteq \Syn(\A/\sigma_i)\cap \left(\bigcap_{i\neq j} I (\sigma_j)\right)
    \]
    and $\sigma_j\subseteq \rho(I_i)$ for all $j\neq i$.
\end{lemma}
\begin{proof}
Let $\sigma_i, \sigma_j \in \mathcal{C}$ with $i \neq j$. Since $\Syn(\A/\sigma) \neq \Syn(\A)$ for every $\sigma \in \mathcal{C}$, for each $i \in [1,s]$ there exists a non-trivial minimal ideal $I_i$ such that 
$$I_i \subseteq \Syn(\A/\sigma_i).$$ 
Observe that for every $u \in I_i \setminus \{0\}$ we have $Q \cdot u \subseteq [p]_{\sigma_i}$ for some $p \in Q$. Moreover, since $I_i I_j = 0$ and $u$ is non-zero, there must exist distinct elements $p, q \in Q \cdot u$ such that 
$$p \cdot v = q \cdot v \quad \text{for all } v \in I_j.$$ 
Hence, $(p,q) \in \rho(I_j) \cap \sigma_i$, so in particular $\rho(I_j) \cap \sigma_i \neq \Delta_\A$. By the minimality of $\sigma_i$, it follows that $\sigma_i \subseteq \rho(I_j)$.  
Now, by Proposition~\ref{morfpos} and the inclusion $\sigma_i \subseteq \rho(I_j)$, we deduce that 
$$I_j \subseteq I(\rho(I_j)) \subseteq I(\sigma_i) \quad \text{for all } i \neq j.$$  
Therefore, for every $i \in [1,s]$ we obtain the claim:  
$$I_i \subseteq \Syn(\A/\sigma_i) \cap \Big( \bigcap_{j \neq i} I(\sigma_j) \Big) 
\quad \text{and} \quad \sigma_j \subseteq \rho(I_i) \text{ for all } j \neq i.$$
\end{proof}

We henceforth set $\tau_i := \rho(I_i)$ where $I_i$ are the non-zero minimal ideals of the previous lemma.
\begin{lemma}\label{336}
With the above conditions: 
    $$\tau:=\bigcap_{i=1}^s \tau_i = \Delta_\A$$
\end{lemma}

\begin{proof}
Assume that $\tau \neq \Delta_{\A}$. Then, there must exist some $i \in \{1,\ldots,s\}$ such that $\sigma_i \subseteq \tau$.  
Applying once more the Galois connection, we obtain  
\[
\begin{split}
\sigma_i \subseteq \tau \subseteq \tau_i = \rho(I_i) 
& \;\Rightarrow\; I_i \subseteq I(\rho(I_i)) \subseteq I(\sigma_i) \\[4pt]
& \;\Rightarrow\; I_i \subseteq \Syn(\A/\sigma_i) \cap I(\sigma_i) = 0,
\end{split}
\]
where the last equality follows from Lemma~\ref{lem: annh}.  
This yields a contradiction.
\end{proof}
The following lemma is a step toward a positive solution to Problem~\ref{prob: quotient of semisimple}.
\begin{lemma}\label{337}
With the above conditions, $\A/\tau_i$ is semisimple for every $i\in\{1,\ldots,s\}$.
\end{lemma}

\begin{proof}
    Assume that $\A/\tau_i$ is not semisimple, and let $I_i$ be the minimal non-zero ideal as in Lemma~\ref{335}. 
    We consider two mutually exclusive cases:
    \begin{itemize}
        \item Assume that $\rad(\A/\tau_i) \cdot I_i \neq 0$.  
By the minimality of $I_i$, it follows that $I_i \subseteq \rad(\A/\tau_i)$.  
Since $\A$ is semisimple and $I_i$ is $0$-minimal, we deduce that 
\[
I_i = I_i^{m} \subseteq \rad(\A/\tau_i)^{m} \subseteq \Syn(\A/\tau_i)
\]
for every integer $m$ greater than or equal to the nilpotency order of $\A/\tau_i$. Suppose that $\Delta_\A \neq \sigma_i\cap\tau_i$, by the minimality of $\sigma_i$ we deduce $\sigma_i\subseteq \tau_i$ which together with Proposition~\ref{morfpos} we conclude that 
    \[I(\sigma_i) \supseteq I(\tau_i) = I(\rho(I_i)) \supseteq I_i\]
  Hence, $I_i\subseteq \Syn(\A/\sigma_i)\cap I(\sigma_i) = 0$ by Lemma~\ref{lem: annh} and the fact that $\A$ is semisimple. This, however, contradicts the fact that $I_i$ is non-trivial. Therefore, we deduce $\tau_i\cap\sigma_i = \Delta_\A$. Now, using the previous inclusion $I_i\subseteq \Syn(\A/\tau_i)$, Lemma~\ref{335} and Proposition~\ref{intsin}, we have:
$$I_i\subseteq \Syn(\A/\sigma_i)\cap\Syn(\A/\tau_i) = \Syn(\A/\tau_i\cap\sigma_i) = \Syn(\A)$$
where the last equality follows from the condition $\tau_i\cap\sigma_i = \Delta_\A$. The previous inclusion, however, contradicts the fact that $I_i$ is non-trivial.
    
    \item \ Thus, we may assume $\rad(\A/\tau_i) \cdot I_i = 0$. In this case, consider a word $u\in\rad(\A/\tau_i)\setminus\Syn(\A/\tau_i)$ ($\A/{\tau_i}$ is not semisimple!). Then, we have $u\cdot I_i = 0$, and $Q\cdot u \not\subseteq [p]_{\tau_i}$ for every $p\in Q$, since $u\notin \Syn(\A/\tau_i)$. Therefore, there are distinct states $p,q\in Q$ belonging to different $\tau_i$-classes $ [p]_{\tau_i}\neq [q]_{\tau_i}$ with the property that $p\cdot v=q\cdot v$ for all $v\in I_i$, hence $(p,q)\in \rho(I_i)\setminus\tau_i$, that is $ \tau_i \neq \rho(I_i)$, a contradiction.
    \end{itemize}
The previous cases have all led to a contradiction, which arises from our initial assumption that $\A/\tau_i$ is not semisimple. Therefore, this assumption must be false, which proves that $\A/\tau_i$ is semisimple and completes the proof of the lemma.
\end{proof}
We are now in position to prove the following  main result.
\begin{teo}\label{simsemquas}
    Assume that the Černý conjecture is solved for simple and strongly connected quasi-simple DFA. Then it holds in general for strongly connected semisimple DFA.
\end{teo}

\begin{proof}
    Let $\A$ be a semisimple DFA. Assume that $\A$ is not simple or quasi-simple. We prove the statement by induction on the number of states $n:=|Q|$. 
    \begin{itemize}
        \item If $\A$ admits a monolith $\sigma$, then by Proposition~\ref{congminsem} we have that $\A/\sigma$ is semisimple and $\Syn(\A)=\Syn(\A/\sigma)$. Then by induction hypothesis we have that $\exists u \in\Syn(\A/\sigma) = \Syn(\A)$ such that $|u|\leq (n-2)^2 < (n-1)^2$.
        \item Assume that $\A$ admits $k$ minimal non-trivial atoms, namely $\mathcal{C}=\{\sigma_1,\ldots,\sigma_k\}$ with $k>1$. If $\Syn(\A/\sigma_i) = \Syn(\A)$ for some $i \in \{1,\ldots,k\}$, then by Lemma~\ref{synquo} together with the semisimplicity of $\A$ we have:
        $$\rad(A/\sigma_i) = \rad(\A) = \Syn(\A) = \Syn(\A/\sigma_i)$$
        and thus the quotient $\A/\sigma_i$ is semisimple. By the induction hypothesis, there exists 
$u \in \Syn(\A/\sigma_i) = \Syn(\A)$ such that 
$|u| \leq (n-2)^2$, and we are done.  
Hence, we may assume that $\Syn(\A/\sigma_i) \neq \Syn(\A)$ for every $i$. Consider the congruences
\[
\tau := \bigcap_{i=1}^{k-1} \tau_i, \qquad \rho := \tau_k,
\]
where $\tau_i = \rho(I_i)$ as defined in Lemma~\ref{335}.  
By that lemma, we have $\tau_i \supseteq \sigma_k$ for every $i \in \{1,\ldots,k-1\}$, and thus $\tau \neq \Delta_\A$.  
Furthermore, Proposition~\ref{intsems} combined with Lemma~\ref{337} implies that $\A/\tau$ is semisimple; the same lemma also guarantees that $\A/\rho = \A/\tau_k$ is semisimple.  
By the induction hypothesis, let $u \in \Syn(\A/\rho)$ and $v \in \Syn(\A/\tau)$ be a Černý reset words in their respective quotients.  
We have $Q \cdot v \subseteq [p]_\tau$ for some $p \in Q$.  
Assume that $\tau$ admits a $1$-class, i.e., $[q]_\tau = \{q\}$ for some $q \in Q$.  
By the strong connectedness of $\A/\tau$ (which comes from the same property for $\A$), there exists $z \in \Sigma^*$ with $|z| \leq n-2$ such that 
$[p]_\tau \cdot z = [q]_\tau = \{q\}$.  
Hence $vz \in \Syn(\A)$ with 
\[
|vz| \leq (n-2)^2 + (n-2) < (n-1)^2.
\]
A similar argument excludes the case in which $\rho$ admits a $1$-class. Therefore, we may assume that every $\tau$-class and every $\rho$-class contains at least two elements.  
This yields $|Q/\tau|, |Q/\rho| \leq n/2$, and hence 
$|u|, |v| \leq (n/2 - 1)^2$, which implies 
\[
|uv| \leq 2 \bigl(n/2 - 1\bigr)^2 < (n-1)^2.
\]  

Finally, by Lemma~\ref{336} we have $\tau \cap \rho = \Delta_\A$, and hence by Proposition~\ref{intsin} we obtain 
\[
uv \in \Syn(\A/\tau) \cap \Syn(\A/\rho) = \Syn(\A),
\]
which completes the inductive step and thus the proof.
    \end{itemize}
\end{proof}

We can now conclude by proving our main result:

\begin{teo}\label{corfinale}
    If the Černý conjecture holds for simple, strongly connected quasi-simple and strongly connected radical automata, then it holds in general. Furthermore, any extremal automaton lies in one of the above mentioned classes.
\end{teo}

\begin{proof}
    Combining Theorem~\ref{teosemrad}, Theorem~\ref{simsemquas} and Proposition~\ref{prop: str conn} we immediately deduce the reduction in the proof of the Černý conjecture. Now, observe that in the proofs of Theorem~\ref{teosemrad} and Theorem~\ref{simsemquas} we are always able to construct synchronizing words that strictly satisfy the Černý bound, and thus any extremal automata must lie in one of our classes.  
\end{proof}

Recall that the only known examples of extremal automata are the \v{C}ern\'y automata and a few sporadic ones. It is straightforward to verify that all of these belong to the class of simple automata, which leads us to the following conjecture:

\begin{conj}
    Any extremal automaton is simple.
\end{conj}

We conclude this work by presenting a diagram depicted in Fig.~\ref{fig:inclusion-diagram} showing how the different classes of automata introduced in this paper are related by inclusion. We also included references to examples that show these inclusions are proper and non-trivial. The only missing example (which, as far as we know, is still unknown) is that of a simple but non-irreducible automaton, which suggests that these two classes might actually coincide. We also believe that the class of simple (respectively, irreducible) automata has measure one with respect to the uniform distribution. More precisely, we pose the following conjecture.

\begin{opbm}
Let $n \in \mathbb{N}$, and let $\mathcal{A}_n$ be a deterministic finite automaton with $n$ states, chosen uniformly at random among all complete deterministic automata over a fixed finite alphabet of size greater than $1$. Is it true that the probability that $\mathcal{A}_n$ is simple tends to $1$ as $n \to \infty$? What can be said about the corresponding statement for irreducibility?
\end{opbm}

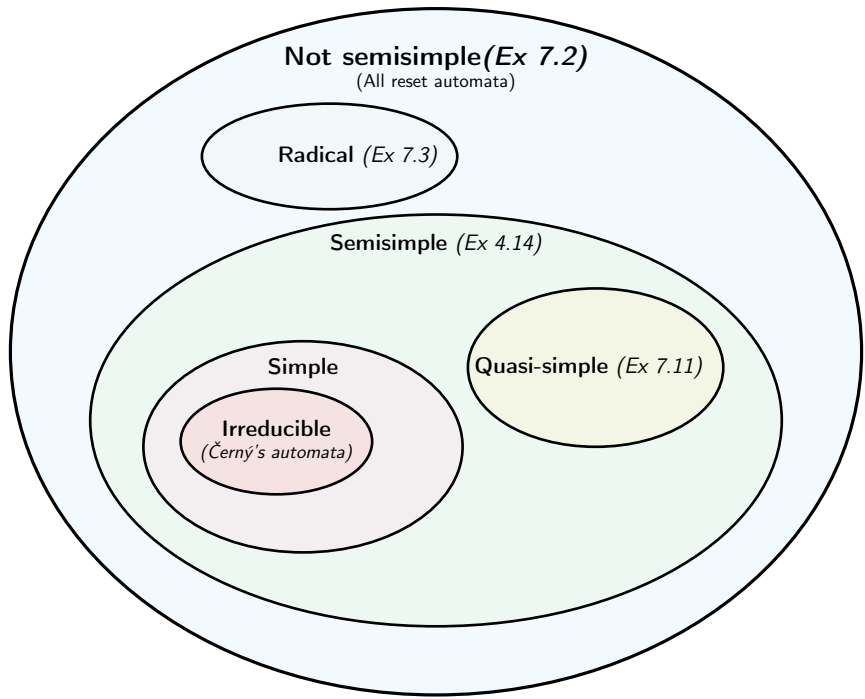
\begin{figure}[ht]
    \centering
    \begin{tikzpicture}[scale=0.7, transform shape, every node/.style={font=\sffamily,align=center}]

% Colori pastello leggermente schiariti
\definecolor{sync}{RGB}{210,235,255}
\definecolor{semi}{RGB}{225,245,215}
\definecolor{quasi}{RGB}{255,240,205}
\definecolor{simple}{RGB}{255,225,235}
\definecolor{irred}{RGB}{250,210,210}
\definecolor{radical}{RGB}{240,240,240}

% Ellisse più esterna: Synchronizing
\draw[fill=sync, fill opacity=0.25, line width=1.2pt]
  (0,1.3) ellipse [x radius=8cm, y radius=6.5cm];

% Ellisse Semisimple
\draw[fill=semi, fill opacity=0.3, line width=1pt]
  (0,0) ellipse [x radius=6.5cm, y radius=3.9cm];

% Ellisse Quasi-simple
\draw[fill=quasi, fill opacity=0.35, line width=1pt]
  (3,1) ellipse [x radius=2.4cm, y radius=1.5cm];

% Ellisse Simple
\draw[fill=simple, fill opacity=0.35, line width=1pt]
  (-2.5,-0.5) ellipse [x radius=3cm, y radius=2cm];

% Ellisse Irred.
\draw[fill=irred, fill opacity=0.4, line width=1pt]
  (-3,-0.4) ellipse [x radius=1.8cm, y radius=1cm];

% Ellisse Radical 
\draw[fill=radical, fill opacity=0.3, line width=1pt]
  (-2,5) ellipse [x radius=2.4cm, y radius=1cm];

% Etichette (leggermente più grandi per leggibilità)
\node[above=1mm] at (0,6) {\Large\textbf{Not semisimple\textit{(Ex 7.2) }}\\(All reset automata)};
\node[above=1mm] at (0,2.9) {\large\textbf{Semisimple} \textit{(Ex 4.14)}};
\node[right=2mm] at (-3.3,5) {\large\textbf{Radical} \textit{(Ex 7.3)}};
\node[right=1mm] at (0.5,1) {\large\textbf{Quasi-simple} \textit{(Ex 7.11)}};
\node at (-2.5,1) {\large\textbf{Simple}};
\node at (-3,-0.4) {\large\textbf{Irreducible}\\ \textit{(Černý's automata)}};

\end{tikzpicture}
    \caption{Venn diagram of the main automata classes considered in the paper.}
    \label{fig:inclusion-diagram}
\end{figure}

\section*{Acknowledgments}
We are grateful to the anonymous referees for their careful reading and valuable feedback, which helped us improve the clarity of the paper.

\printbibliography[title={References}]

\end{document}